\documentclass{article}

\usepackage{microtype}
\usepackage{graphicx}
\usepackage{subcaption}
\usepackage{booktabs} 

\usepackage{hyperref}

\usepackage[preprint]{icml2026}

\usepackage{amsmath}
\usepackage{amssymb}
\usepackage{mathtools}
\usepackage{amsthm}

\usepackage{graphicx}
\usepackage{url}            
\usepackage{booktabs}       
\usepackage{amsfonts}       
\usepackage{nicefrac}       
\usepackage{microtype}      
\usepackage{xcolor}         
\usepackage{multirow}
\usepackage{balance}
\usepackage{xspace}
\usepackage{enumitem}
\usepackage{wrapfig}
\usepackage{colortbl}
\usepackage{hyperref}
\usepackage{amsmath}
\usepackage{mathtools}
\usepackage{amsthm}
\usepackage{braket}

\newcommand{\ourmethod}{TF-LLMER\xspace}
\newcommand{\bs}{\boldsymbol}
\newcommand{\mc}{\mathcal}
\newcommand{\mb}{\mathbf}

\newcommand{\softmax}{\mathop{\text{softmax}}}

\newcommand{\R}{\mathbb{R}}

\newcommand{\llmname}[1]{{\fontfamily{pcr}\selectfont {#1}}\xspace} 

\usepackage[capitalize,noabbrev]{cleveref}

\theoremstyle{plain}
\newtheorem{theorem}{Theorem}[section]
\newtheorem{proposition}[theorem]{Proposition}
\newtheorem{lemma}[theorem]{Lemma}

\theoremstyle{definition}
\newtheorem{definition}[theorem]{Definition}
\newtheorem{assumption}[theorem]{Assumption}
\theoremstyle{remark}

\usepackage[textsize=tiny]{todonotes}

\icmltitlerunning{Break the Optimization Barrier of LLM-Enhanced Recommenders: A Theoretical Analysis and Practical Framework}

\begin{document}

\twocolumn[
  \icmltitle{Break the Optimization Barrier of LLM-Enhanced Recommenders: \\A Theoretical Analysis and Practical Framework}



  \icmlsetsymbol{equal}{*}

  \begin{icmlauthorlist}
    \icmlauthor{Zhangchi Zhu}{ecnu}
    \icmlauthor{Wei Zhang}{ecnu,inov}
  \end{icmlauthorlist}

  \icmlaffiliation{ecnu}{East China Normal University, Shanghai, China}
  \icmlaffiliation{inov}{Shanghai Innovation Institute, Shanghai, China}

  \icmlcorrespondingauthor{Wei Zhang}{zhangwei.thu2011@gmail.com}

  \icmlkeywords{Sequential Recommendation, Large Language Models, Graph Signal Processing}

  \vskip 0.3in
]



\printAffiliationsAndNotice{}  

\begin{abstract}
Large language model (LLM)-enhanced recommendation models inject LLM representations into backbone recommenders to exploit rich item text without inference-time LLM cost. However, we find that existing LLM-enhanced methods significantly hinder the optimization of backbone models, resulting in high training losses that are difficult to reduce. To address it, we establish a comprehensive theoretical analysis of local optimization curvature and identify two key causes: 1) large norm disparity and 2) semantic–collaboration misaligned angular clustering of LLM representations.
Guided by these insights, we propose \textbf{T}raining-\textbf{F}riendly \textbf{LLM-E}nhanced \textbf{R}ecommender (\textbf{\ourmethod}), a lightweight framework with two key components. First, we highlight the necessity of \emph{item embedding normalization} to eliminate norm-driven instability and achieve provable control over optimization conditioning. Second, we introduce \emph{Rec-PCA}, a recommendation-aware dimensionality reduction method that injects collaborative structure into the representation transformation to resolve semantic–collaboration misaligned angular clustering. It jointly optimizes semantic information retention and alignment with an item–item co-occurrence graph constructed from interaction histories. The graph captures collaborative structure, and alignment is promoted by penalizing total variation over the graph.
Both theory and extensive experiments demonstrate that \ourmethod significantly outperforms state-of-the-art methods. Our code is available at \url{https://github.com/woriazzc/TF-LLMER}.
\end{abstract}

\begin{figure}
\centering
  \includegraphics[width=\linewidth]{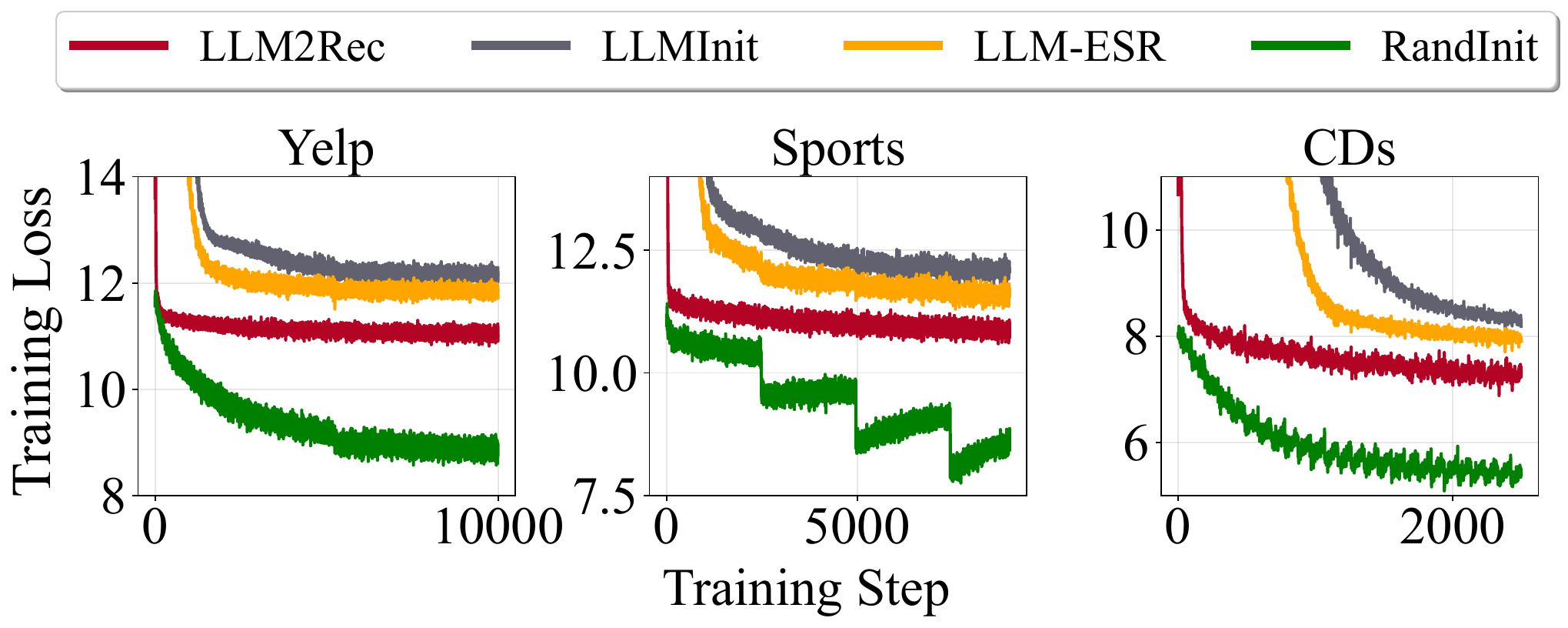}
    \vspace{-1.5em}
  \caption{Training loss of several methods, including the standard randomly initialized model (denoted as RandInit) and state-of-the-art LLM-enhanced methods. The backbone is GRU4Rec. Results on other backbones are provided in Figure~\ref{fig:loss_methods_all} in the appendix.}
  \label{fig:loss_methods_main}
    \vspace{-1.5em}
\end{figure}

\section{Introduction}\label{sec:intro}


Recent progress on large language models (LLMs) has led to a prominent line of \emph{LLM-enhanced recommender systems}. They leverage LLM-derived textual representations as item embeddings (or as initialization of item embeddings) of conventional recommendation backbone models, aiming to boost performance while avoiding inference-time LLM costs. This paradigm typically follows two steps: 1) use LLMs to extract appropriate semantic representations from the textual attributes of items, 2) transform these LLM representations to fit the embedding layer of a backbone recommendation model and retrain the backbone after the injection~\cite{shi2025matters}. For example, LLMInit~\cite{zhang2025llminit} uses the last tokens' hidden states as LLM representations. Then it initializes item embeddings using dimensionally reduced LLM representations. LLM-ESR~\cite{liu2024llm} transforms LLM representations using both PCA and adaptors, and combines them through cross-attention. LLM2Rec~\cite{he2025llm2rec} adapts LLMs to produce recommendation-friendly representations through collaborative supervised fine-tuning, followed by adaptor-based transformation into item embeddings.

However, existing methods focus solely on extracting better LLM representations, neglecting the effects after injecting them into the embedding layers of recommendation models. This very fact contains a significant problem. Specifically, we find that existing methods consistently encounter extremely high recommendation losses that are difficult to decrease when retraining recommendation models after injecting LLM representations, suggesting an optimization barrier. As shown in Figure~\ref{fig:loss_methods_main}, compared to the standard randomly initialized GRU4Rec, training losses of all LLM-enhanced models cease to decrease rapidly and stabilize at significantly higher values. Clearly, their optimization is hindered. However, existing methods have never explored how injecting LLM representations into backbone recommendation models affects their training, let alone identified this severe negative impact and discussed how to address it.

To address this oversight, we adopt a representation-level optimization perspective and analyze the optimization curvature of the training loss with respect to the sequence's representation. We establish a theoretical framework that uses the Hessian matrix's condition number to quantify the representation trainability. Our analysis reveals two fundamental and practically decisive obstacles in training the backbone of LLM-enhanced recommendation models: (i) \textbf{large norm disparities in LLM-derived item embeddings} that can make the induced curvature arbitrarily ill-conditioned, and (ii) \textbf{semantic-collaboration misaligned angular clustering} that blurs the separation among training-critical items and undermines effective optimization. These issues directly explain the unexpectedly poor backbone training behavior after injecting LLM representations.

To mitigate the norm-driven issue, we identify a necessary remedy: \textbf{item embedding normalization when computing item logits}. Existing LLM-enhanced methods treat LLM embeddings as plug-in item features and thus overlook the importance of normalization, leaving norm disparity to propagate into severe curvature ill-conditioning during backbone training. In contrast, normalization removes norm-induced instability and provides provable control over conditioning. We verify the necessity of normalization in LLM-enhanced recommenders both theoretically and empirically.

As for the second issue, it refers to the angular clustering of representations of recommendation-critical items, arising from the mismatch between semantic relevance within LLMs and the collaborative relevance required by recommendation tasks.
Existing methods introduce extra training when extracting LLM representations to adapt them to recommendation tasks. It has three fatal drawbacks: 1) substantial additional training overhead, 2) inappropriate training both damages LLM knowledge and fails to satisfy recommendation needs, and 3) gains are likely diluted by subsequent transformations required to inject LLM representations into item embeddings. We therefore propose \textbf{Rec-PCA}, a recommendation-aware dimensionality reduction method that injects collaborative information during representation-to-embedding transformation. It achieves semantic retention together with recommendation alignment enforced via a graph total-variation penalty on an item–item co-occurrence graph. A hyperparameter controls their trade-off. Unlike prior training-based adaptation, Rec-PCA 1) requires no training, only a PCA-like preprocessing. 2) A hyperparameter controls the trade-off, making the adaptation controllable and easy to tune. 3) Because it is applied directly in the transformation step, its effect is not diluted.

In summary, our contributions include:
\newline
$\bullet$ We develop a comprehensive theoretical framework based on representation-level induced curvature. Based on it, we show that injecting LLM representations into conventional recommenders significantly hinders their training.
\newline
$\bullet$ We theoretically show that unnormalized LLM-derived item embeddings lead to unbounded ill-conditioning, and item embedding normalization has long been overlooked yet remains essential to training LLM-enhanced recommenders. We validate it theoretically and empirically.
\newline
$\bullet$ We propose \textbf{Rec-PCA}, a collaboration-aware method that introduces a tunable mechanism to trade off semantic information retention and recommendation-task alignment during dimensionality reduction. We empirically and theoretically verify its effect on facilitating training.
\newline
$\bullet$ We conduct extensive experiments on three public datasets and three backbones to demonstrate the superiority of our method. We also empirically show that our method is compatible with existing LLM-enhanced methods.

\section{Related Work}

\subsection{Graph-based Recommender Systems}

Inspired by graph-based recommender systems, our Rec-PCA endows representations with collaborative structures by reducing their total variation across co-occurrence graphs.
Graph-based recommender systems~\cite{wu2022graph,ricci2021recommender,wang2021survey} conduct recommendations with graph signal processing. NGCF~\cite{wang2019neural} and LightGCN~\cite{he2020lightgcn} are the first to propose a graph-based recommendation framework. SR-GNN~\cite{wu2019session} introduces GNN into session-based recommendation by constructing an item-item co-occurrence graph from sequential interaction data. Recently, some works have analyzed graph models from a spectral perspective. They define low-frequency as representations with smaller total variation. GF-CF~\cite{shen2021powerful} shows that traditional graph-based recommenders are essentially low-pass filters. That is, they reduce the total variation of item representations on the item graph, making representations better match the graph. GDE~\cite{peng2022less} validates that low-frequency components that have lower total variation are much more important than high-frequency components. 

\subsection{LLM-Enhanced Recommender Systems}

LLM-enhanced recommendation models~\cite{liu2025large,zhang2024dual,zhang2024id} typically involve two steps~\cite{shi2025matters}: 1) extract textual representations from LLM. Some methods design training strategies at this step to adapt them to the recommendation task. 2) Transform these LLM representations to fit the embedding layer of a backbone recommendation model, and retrain the backbone after the injection. As pioneers, AlphaRec~\cite{sheng2024language} directly uses adaptors to transform LLM representations, treating them as item embeddings. LLMInit~\cite{zhang2025llminit} initializes item embeddings using dimensionally reduced LLM representations. Building upon these two straightforward methods, a great deal of research has been generated.
LLM-ESR~\cite{liu2024llm} transforms LLM representations into item embeddings through both adaptors and PCA-based initialization, and fuses them via cross attention. Pushing further, LLMEmb~\cite{liu2025llmemb} adapts LLMs to produce recommendation-friendly representations via contrastive fine-tuning and applies adaptors after PCA-reduced representations to integrate them into recommendation models. AlphaFuse~\cite{hu2025alphafuse} decomposes the language-embedding space and trains ID embeddings in the null space of LLM representations. Finally, LLM2Rec~\cite{he2025llm2rec} adapts LLMs to recommendation tasks through collaborative supervised fine-tuning. Then, it uses adaptors to transform LLM representations into item embeddings.

However, existing methods focused solely on extracting better LLM representations, neglecting the implications of injecting them into the embedding layers of recommendation models. Yet this very oversight conceals the most critical problem: LLM representation injection severely impedes backbone model training. Our work provides a detailed theoretical analysis and a practical framework for this issue.

\section{Prelimilary}

\subsection{Sequential Recommendation}
This work focuses on the sequential recommendation task.
Let $\mc I$ denote the item set and $|\mc I|$ denote the number of items. For a sample $(\mb q, y)$, consisting of the historical interaction sequence $\mb q$ and the ground-truth next item $y$, sequential recommendation models first represent $\mb q$ as $\bs h\in\R^d$, where $d$ is the dimension of embeddings. Then, they compute the logits on all items as the dot product of $\bs h$ and all items' embeddings $\mathbf{E}\in\R^{|\mc I|\times d}$, i.e., $\bs s=\mb E\bs h\in\R^{|\mc I|}$.

Following previous work~\cite{sun2019bert4rec,xu2024understanding,klenitskiy2023turning}, we use Cross-Entropy loss to retrain the backbone recommendation model after injection:
    \vspace{-1.5em}
\begin{align}
    \ell (\bs s, y)=-s_y+\log\sum_{j\in\mc I}e^{s_j}.\label{eq:loss}
\end{align}
    \vspace{-1.5em}

\subsection{Graph Signal Processing}

Our Rec-PCA augments PCA with graph signal processing. Given a graph $\mathcal{G}=(\mathcal{V},\mathcal{E})$ with $\mathcal{V}$ being the node set and $\mathcal{E}$ being the edge set, we denote $|\mathcal{V}|$ and $|\mathcal{E}|$ the number of nodes and the number of edges, respectively. Let $\mathbf{A}\in \mathbb{R}^{|\mathcal{V}|\times |\mathcal{V}|}$ be the adjacency matrix of $\mathcal{G}$ and $\mathbf{D}=\text{Diag}(\mathbf{A}\mathbf{1})$ be the degree matrix with $\mathbf{1}\in\mathbb{R}^{|\mathcal{V}|}$ being the all-one vector. The symmetric normalized graph Laplacian matrix is given by $\mathbf{L}=\mathbf{I}-\mathbf{D}^{-1/2}\mathbf{A}\mathbf{D}^{-1/2}$.
Graph signal processing provides tools for describing the relationship between graph signals and graphs. Specifically, given a graph $\mc G$ and a graph signal $\bs x\in\R^{|\mc V|}$, the \textit{total variation (TV)}~\cite{ortega2022introduction} measures the degree of matching between them. Formally, it is defined as
    \vspace{-0.5em}
\begin{align}
    TV(\bs{x};\mc G)=\sum_{i,j\in\mc V}\mathbf{A}_{i,j}(\bs{x}_i-\bs{x}_j)^2=\bs{x}^\top\mb{L}\bs{x}.
\end{align}
A smaller $TV(\bs{x};\mc G)$ indicates that the values of adjacent nodes in graph $\mc G$ have closer values in the signal $\bs x$. In other words, $\bs x$ matches $\mc G$ more closely. As shown in many works on graph-based recommender systems~\cite{guo2023manipulating,peng2022less,peng2022svd}, signals with smaller $TV(\bs{x};\mc G)$ often have a greater impact on recommendation performance.

\subsection{Model Optimization}

To theoretically analyze model training, we draw on relevant research in optimization theory.
The mainstream approach for evaluating the difficulty of a model optimization process involves analyzing the Hessian matrix~\cite{nesterov2013introductory}.
However, the full Hessian of recommendation models is intractable due to their higher-order interactions within attention/normalization modules~\cite{wu2020dissecting}.
We thus analyze the representation-level curvature to characterize the local anisotropy of loss in the subspace of reachable representation perturbations~\cite{martens2010deep,botev2017practical}. Specifically, we note that the sequence representation $\bs h$ is the decision representation: all ranking scores are functions of $\bs h$ through $\bs s=\mb E\bs h$. Therefore, the Hessian matrix of $\bs h$ that describes the local curvature of $\ell$ with respect to $\bs h$ provides an actionable description of the anisotropy.

By rewriting $\ell(\bs s,y)$ in Eq.(\ref{eq:loss}) equivalently as $\ell(\bs h)=\log\softmax(\mb E\bs h)_y$, the Hessian with respect to $\bs h$ is $\mb H_h:=\nabla_{\bs h}^2\ell(\bs h)=\mb E^\top\mb H_s\mb E$, where $\mb H_s$ is the Hessian matrix w.r.t. the logits $\bs s$.
The formula of $\mb H_s$ and derivation of Hessian matrices are in Appendix~\ref{proof:h}.
Then, we analyze the condition number of the Hessian matrix $\mb H_h$ defined as follows.
\begin{definition}[condition number~\cite{belsley2005regression}]
    The condition number of $\text{ }\mb H_h$ is
    \vspace{-0.5em}
    \begin{align}
        \kappa(\mb H_h)=\frac{\lambda_{\max}(\mb H_h)}{\lambda_{\min}(\mb H_h)},
    \end{align}
    where $\lambda_{\max}(\cdot)$ and $\lambda_{\max}(\cdot)$ denote the largest and smallest eigenvalues. The condition number of the Hessian matrix reflects the difficulty of optimization. In short, the larger the condition number, the more difficult the optimization. 
\end{definition}

However, due to the long-tail nature of recommender systems, most items are predicted to have near-zero interaction probability. The curvature is effectively concentrated on only a small subset of effective items (i.e., the positive and a small set of hard negatives). Therefore, analyzing the entire item universe yields $\lambda_{\min}(\mb H_h)\approx 0$, rendering the condition number of $\mb H_h$ undefined. To address this, we define an effective subspace that actually affects the optimization.

\begin{definition}[Effective Subspace]
     We define the effective subspace $\mc U\subset \R^{|\mc I|}$ as the subspace of the logit space spanned by the support of the gradient along the optimization trajectory. Let $m:=\text{dim}(\mc U)$, which in practice corresponds to the number of effective items, i.e., the positive item plus a small number of hard negatives, and $m\ll|\mc I|$.
    We define $\mc P_{\mc U}$ as the orthogonal projection operator from the entire logits space $\R^{|\mc I|}$ to $\mc U$.
\end{definition}

Then, we assume that the curvature neither collapses nor explodes in this effective subspace $\mc U$.
\begin{assumption}[Local Non-saturation Region]\label{ass:effective}
    We assume that $\text{ }\mb{H_s}$ stays in a local non-saturation region:
    \vspace{-0.5em}
    \begin{align*}
    \forall \bs v\in\mc U, \alpha\|\bs v\|^2\le \bs v^\top\mb H_s\bs v\le \beta\|\bs v\|^2, \text{for some } 0<\alpha\le \beta.
    \end{align*}
\end{assumption}
    \vspace{-0.5em}

Next, we provide updated notation regarding embeddings within our effective subspace. We denote $\mb E_{\mc U}:=\mc P_{\mc U}\mb E=\{\bs e_i\}_{i=1}^m$ as the item embeddings projected into $\mc U$. We also define $\hat{\bs e}_i:=\bs e_{i}/\|\bs e_i\|$ and $\hat{\mb E}_{\mc U}:=\{\hat{\bs e}_i\}_{i=1}^m$ as the normalized item embeddings. Finally, we discard items that contributed little to training and restrict our analysis to $\mc U$ by updating $\mb H_s$ to $\mb H_{s\mc U}:=\mc P_{\mc U}\mb H_s\mc P_{\mc U}^\top$ and $\mb H_h$ to $\mb H_{h\mc U}:=\mb E_{\mc U}^\top\mb H_{s\mc U}\mb E_{\mc U}$.

\section{Existing LLM-enhanced Methods Induce Ill-conditioned Loss Geometry}\label{sec:ill}

\begin{figure}
\centering
  \includegraphics[width=\linewidth]{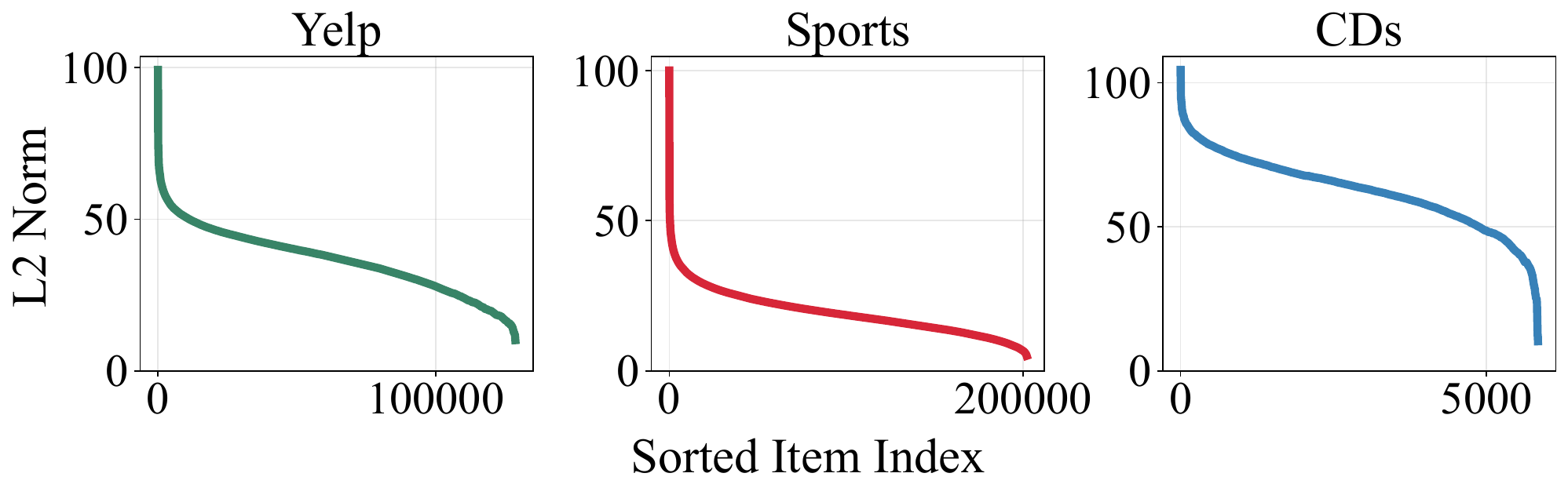}
    \vspace{-1.5em}
  \caption{The magnitudes of the initial item embeddings for all items. Items are sorted in descending order by magnitude. The initial item embeddings are provided by LLM2Rec.}
  \label{fig:norms}
    \vspace{-1.0em}
\end{figure}



Since we study model optimization through Hessian matrix's condition number, we introduce the following theorem that gives the upper bound of $\mb H_{h\mc U}$'s condition number.

\begin{theorem}\label{theorem:norm_disparity}
    Let $r_{\max}:=\max_{i=1}^m\|\bs e_i\|$ and $r_{\min}:=\min_{i=1}^m\|\bs e_i\|$ denote the maximal and minimal norms of item embeddings. We obtain 
    \vspace{-1.0em}
    \begin{align}
        \kappa(\mb H_{h\mc U})\le \frac{\beta}{\alpha}\cdot\left(\frac{r_{\max}}{r_{\min}}\right)^2\cdot\kappa(\hat{\mb E}_{\mc U}\hat{\mb E}_{\mc U}^\top).\label{eq:kappa_h}
    \end{align}
\end{theorem}
    \vspace{-1.0em}
The proof is given in Appendix~\ref{proof:norm_disparity}.
This theorem forms the foundation for our entire analysis, demonstrating that the upper bound of the Hessian matrix's condition number is jointly determined by two key factors:
\begin{enumerate}[leftmargin=*]
\item the disparity of the magnitudes of item embeddings, i.e., $(r_{\max}/r_{\min})^2$,
\item the condition number of the cosine similarity matrix for effective item embeddings, i.e., $\kappa(\hat{\mb E}_{\mc U}\hat{\mb E}_{\mc U}^\top)$.
\end{enumerate}
Unfortunately, injecting LLM representations into item embeddings fatally undermines both factors, as detailed below.

\textbf{(i) Norm-driven Ill-conditioning}

Firstly, our theorem demonstrates that the norm disparity across items, measured by $r_{\max}/r_{\min}$, deteriorates the upper bound of $\kappa(\mb H_{h\mc U})$ in a quadratic form. In practice, we observe significant disparity in the norms of LLM-derived item embeddings across different items. Taking LLM2Rec~\cite{he2025llm2rec}, which we currently observe to be the best-performing baseline method, as an example, Figure~\ref{fig:norms} displays the norms of all items' initial embeddings on three datasets. It can be observed that, despite slightly different distributions across datasets, they all unequivocally point to the same conclusion. That is, the magnitudes of initial embeddings for different items show substantial variation, with a significant disparity between the maximum and minimum. This observation diverges significantly from our theoretical pursuits, hindering subsequent model optimization.

\begin{figure}
\centering
  \includegraphics[width=\linewidth]{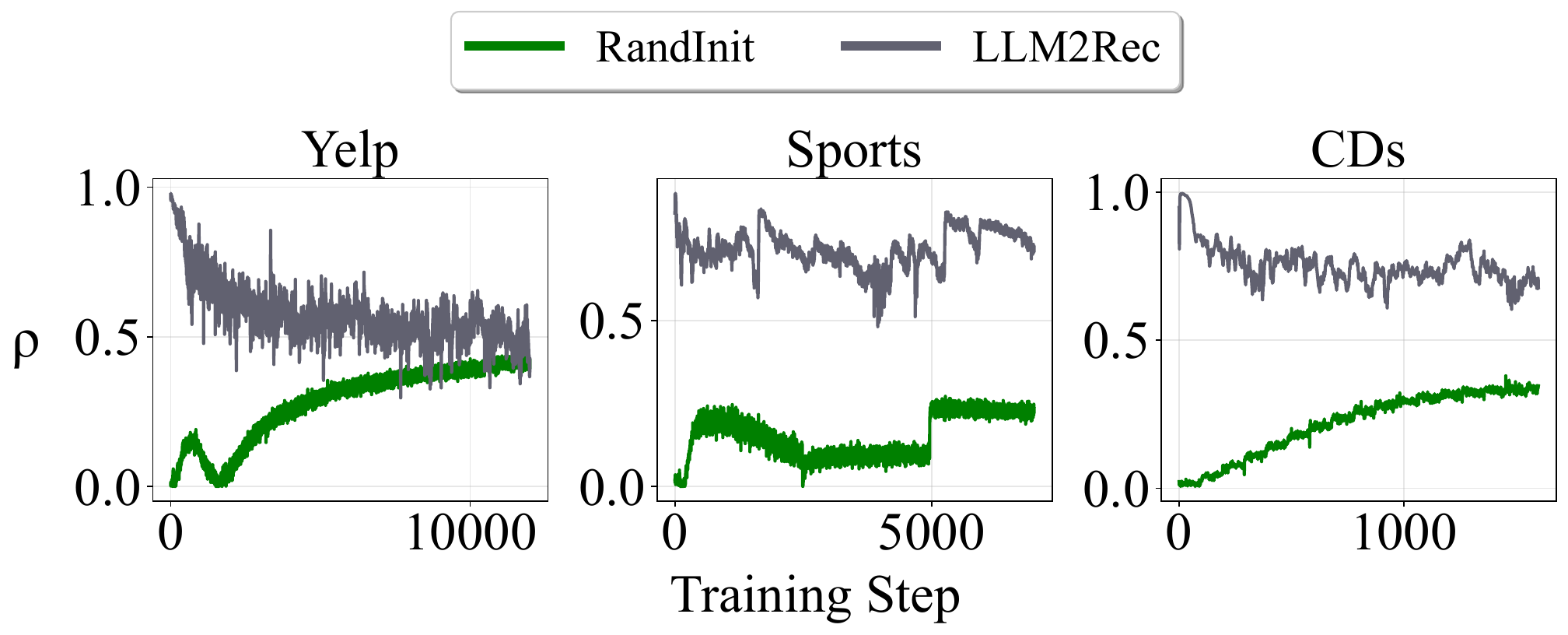}
    \vspace{-1.5em}
  \caption{Training curves of $\rho$ on the training set for randomly initialized model and LLM2Rec across three datasets on SASRec.}
    \vspace{-1.0em}
  \label{fig:rho_method_main}
\end{figure}

\textbf{(ii) Semantic–Collaboration Misaligned Angular Clustering}

Secondly, our theory shows that the larger the condition number of the cosine similarity matrix for item embeddings (i.e., $\kappa(\hat{\mb E}_{\mc U}\hat{\mb E}_{\mc U}^\top)$), the greater the upper bound on the Hessian matrix's condition number. Furthermore, in Theorem~\ref{theorem:norm} and Proposition~\ref{theorem:rho_gcn} of subsequent sections, we prove that $\kappa(\hat{\mb E}_{\mc U}\hat{\mb E}_{\mc U}^\top)$ increases with the maximum similarity between effective items. However, in practice, LLMs contain abundant semantic information unrelated to the recommendation task~\cite{he2025llm2rec,hu2025alphafuse,hu2024enhancing}. This reduces the distinguishability among items crucial for recommendation, significantly amplifying their similarity. Figure~\ref{fig:rho_method_main} illustrates the maximum similarity among effective items throughout the training process for LLM2Rec~\cite{he2025llm2rec} and randomly initialized SASRec. The results suggest that the random initialization is consistently lower than LLM2Rec. It is important to note that the collaborative supervised fine-tuning employed in LLM2Rec already tried to adapt LLM representations to the recommendation task, making it more suitable for recommendation than other LLM-enhanced methods. This result indicates that even when training tasks are carefully designed to adapt LLMs for recommendation, derived item embeddings may still prove unsuitable for recommendation training. This can occur due to factors such as insufficient training or dilution during subsequent representation transformation processes.

\begin{figure}
\centering
  \includegraphics[width=\linewidth]{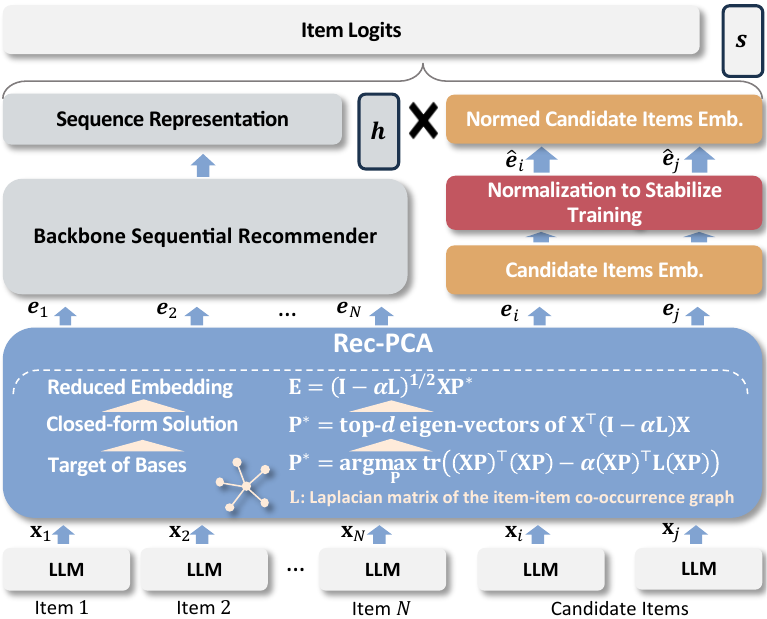}
    \vspace{-1.5em}
  \caption{Framework diagram of our method. It includes two key components: item embedding normalization and Rec-PCA.}
  \label{fig:method}
    \vspace{-1.5em}
\end{figure}

\section{Method}

\subsection{Overview}\label{sec:overview}

Building on the two optimization obstacles identified in the previous section, Figure~\ref{fig:method} summarizes our training-friendly LLM-enhanced framework (read from bottom to top). We first extract an item-level LLM representation $\mathbf{x}_i$ from each item’s textual attributes. The extractor can be any method that produces $\mathbf{x}_i$, making the framework model-agnostic to the choice of LLM or enhancement strategy.

Next, we reduce $\mathbf{x}_i$ via \emph{Rec-PCA} (Section~\ref{sec:seqpca}) to obtain the item embedding $\bs{e}_i$, which is used to initialize the embedding layer of a backbone sequential recommender. Then it produces the sequence representation $\bs h$ and is trained with the recommendation loss (Cross-Entropy). Importantly, during training we compute logits with \emph{normalized candidate item embeddings} $\hat{\bs e}_i$ (Section~\ref{sec:norm}) to prevent norm disparity.

\subsection{Item Embedding Normalization is Necessary for Training LLM-Enhanced Model}\label{sec:norm}

As we have demonstrated, the upper bound of the Hessian matrix's condition number exhibits a quadratic positive correlation with the norm disparity in item embeddings. Existing LLM-enhanced methods introduce significant magnitude disparity into the initial item embeddings by injecting LLM representations, making it challenging to optimize the recommendation model. To address this, we emphasize the need for item embedding normalization, which is neglected by all existing LLM-enhanced methods. Specifically, all LLM-enhanced methods simply inject LLM representations into item embeddings and then proceed without further adjustment. They assume that the backbone's default training strategy can effectively handle models that differ significantly from randomly initialized ones. However, this is clearly erroneous. The substantial change to the model's embedding layer requires careful reconsideration, specifically the normalization of item embeddings.

To justify it, we investigate the impact of magnitude disparity in item embeddings in isolation. To decouple this effect from other factors, we assume that $\kappa(\hat{\mb E}_{\mc U}\hat{\mb E}_{\mc U}^\top)$ is a finite constant in this section.
Then, we can derive the following proposition for the normalized setting from Theorem~\ref{theorem:norm_disparity}.
\begin{proposition}[Conditioning in Normalized Setting]
    Assume that $\kappa(\hat{\mb E}_{\mc U}\hat{\mb E}_{\mc U}^\top)$ is finite. Then, $\kappa(\mb H_{h\mc U})$ is finite if item embeddings are normalized. The upper bound is given by
    \vspace{-0.5em}
    \begin{align}
        \kappa(\mb H_{h\mc U})\le \frac{\beta}{\alpha}\cdot\kappa(\hat{\mb E}_{\mc U}\hat{\mb E}_{\mc U}^\top).
    \end{align}
    \vspace{-1.0em}
\end{proposition}
It demonstrates that, provided other factors do not deteriorate arbitrarily, performing item embedding normalization can ensure a bounded condition number.

\begin{figure}
\centering
  \includegraphics[width=\linewidth]{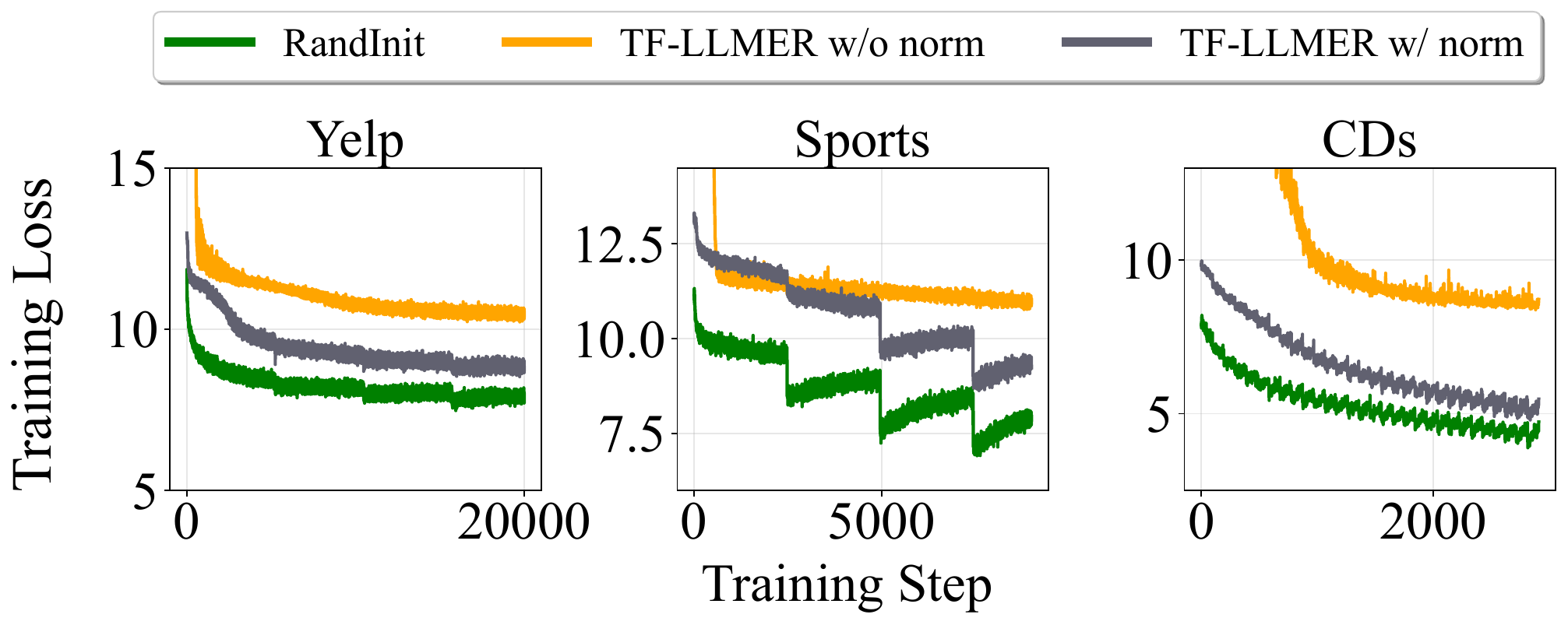}
    \vspace{-1.5em}
  \caption{Training loss of our method without and with normalization. The backbone is SASRec. Results on other backbones are given in Figure~\ref{fig:loss_all} in the appendix.}
  \label{fig:loss_main_norm}
    \vspace{-1.0em}
\end{figure}

In contrast, without item embedding normalization, even if all other factors are bounded, the enormous disparity in embedding norms can still lead to an arbitrarily poor condition number. Specifically, we first present the following theorem, which provides a lower bound on the Hessian matrix's condition number in the unnormalized setting.
\begin{theorem}\label{theorem:unnorm}
The Hessian matrix's condition number is lower bounded by the condition number of the unnormalized similarity matrix. Formally, 
    \vspace{-0.5em}
\begin{align}
    \kappa(\mb H_{h\mc U})\ge \frac{\alpha}{\beta}\cdot\kappa(\mb E_{\mc U}\mb E_{\mc U}^\top).
\end{align}
\end{theorem}
    \vspace{-1.0em}
The proof is given in Appendix~\ref{proof:unnorm}. Note that $\mb E_{\mc U}\mb E_{\mc U}^\top$ here refers to the similarity matrix computed using non-normalized item embeddings.
Based on the lower bound provided in this theorem, we have the following proposition.
\begin{proposition}[Conditioning in Unnormalized Setting]\label{prop:unnorm}
    Assume that $\kappa(\hat{\mb E}_{\mc U}\hat{\mb E}_{\mc U}^\top)$ is finite. There is no upper bound on $\kappa(\mb H_{h\mc U})$ if norms of item embeddings are unconstrained.
\end{proposition}
The proof is in Appendix~\ref{proof:prop_unnorm}. It shows that we cannot theoretically bound the conditioning without normalization.

Therefore, it is foreseeable that the training difficulty will be significantly greater than when normalization is applied. To empirically validate this, Figure~\ref{fig:loss_main_norm} shows the loss curve throughout the entire training process with and without normalization. We observe that without normalization, the loss starts extremely high and its decline ceases very quickly. It also ultimately stabilizes at a level substantially higher than that achieved with normalization. This indicates that the extreme disparity in norms within the initial item embeddings significantly impedes continued training. And, as predicted by our theory, normalization effectively resolves this issue.


\begin{figure}
\centering
  \includegraphics[width=\linewidth]{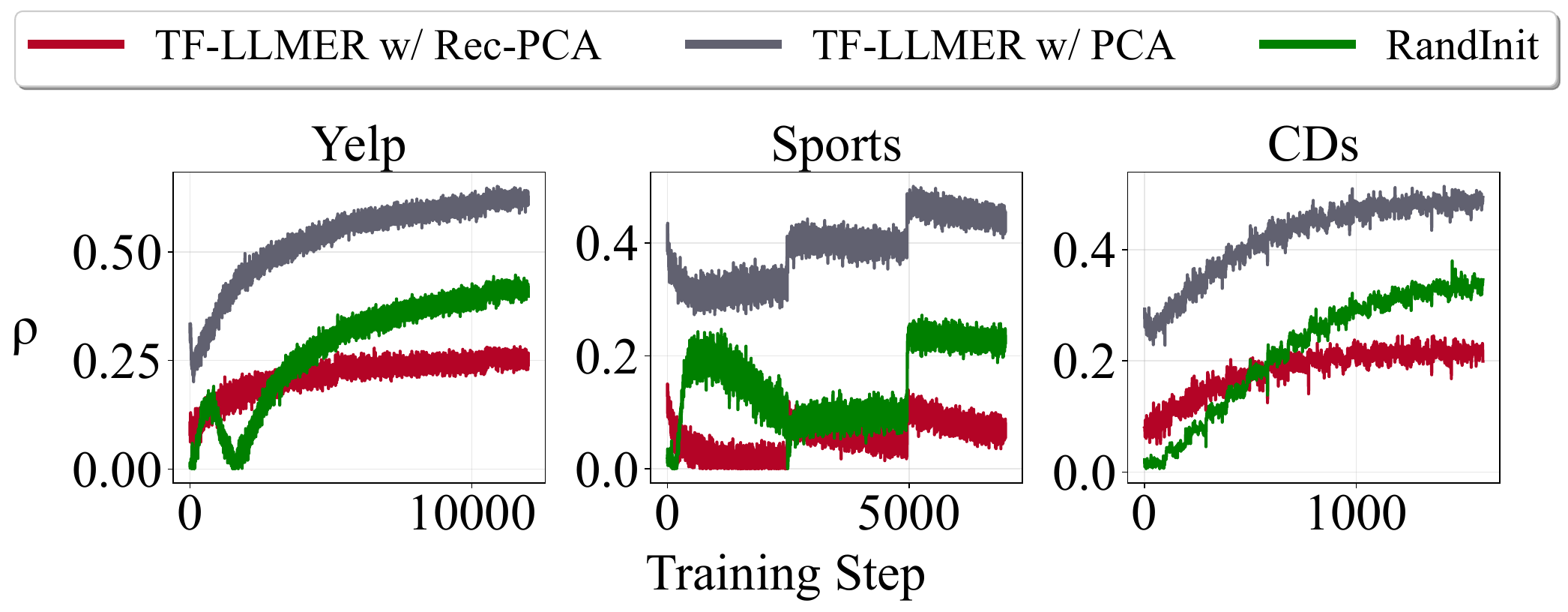}
    \vspace{-1.5em}
  \caption{Training curves of $\rho$ on the training set for our method with Rec-PCA and vanilla PCA, and random initialization. The backbone is SASRec. Full results are in Figure~\ref{fig:rho_all} in the appendix.}
    \vspace{-1.0em}
  \label{fig:rho_main}
\end{figure}

\begin{figure}
\centering
  \includegraphics[width=\linewidth]{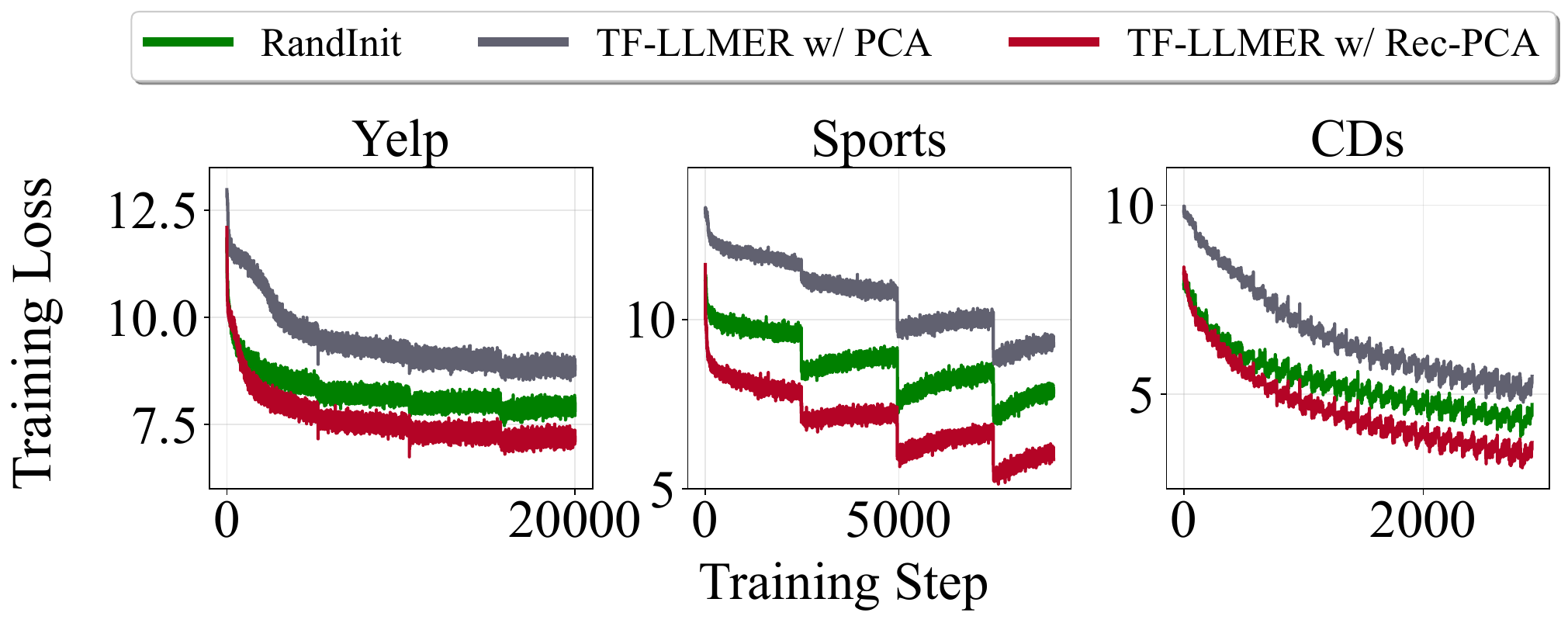}
    \vspace{-1.5em}
  \caption{Training loss of random initialization, our method with Rec-PCA and vanilla PCA. The backbone is SASRec. Full results are in Figure~\ref{fig:loss_all} in the appendix.}
  \label{fig:loss_main}
    \vspace{-1.5em}
\end{figure}

\subsection{Rec-PCA as an Efficient and Effective Method for LLM Representation Adaptation}\label{sec:seqpca}

Another factor affecting the Hessian matrix's condition number is semantic–collaboration misaligned angular clustering.
It arises from the mismatch between the semantic relevance required by LLMs' pre-training tasks and the collaborative relevance demanded by recommendation tasks.
To address this issue, we propose Rec-PCA, a recommendation-aware dimensionality reduction method that only requires one-shot preprocessing.
Specifically, it aims to achieve the following two objectives simultaneously when transforming LLM representations into the initialization of item embeddings: 1) maximizing the information content retained in transformed representations, and 2) maximizing the suitability of transformed representations for recommendation tasks.

For the first objective, we align with standard PCA by maximizing the variance of each dimension after transformation while maintaining zero covariance between dimensions. Let $\mb X\in\R^{|\mc I|\times d_{\text{LLM}}}$ be the representation matrix encoded by LLM and $\mb P\in\R^{d_{\text{LLM}}\times d}$ denote the transformation matrix, where $d_{\text{LLM}}$ is LLM representations' dimension. The first objective is maximizing: $M_1=\text{tr}\left((\mb X\mb P)^\top(\mb X\mb P)\right)$.

For the second objective, we first construct an item-item co-occurrence graph that captures the collaborative structure from the recommendation task. Then we minimize the total variation of the transformed representations on the graph to match it.
Specifically, we construct the item-item co-occurrence graph $\mc G$ by linking adjacent items within all historical interaction sequences. Then the second objective is minimizing: $M_2=TV(\mb X\mb P;\mc G)=\text{tr}\left((\mb X\mb P)^\top\mb L(\mb X\mb P)\right)$, where $\mb L$ denotes the Laplacian matrix of $\mc G$.

Then, we combine these two objectives through the hyperparameter $\alpha$, meaning our final maximization objective is
    \vspace{-1.0em}
\begin{align}\label{eq:target}
    M=M_1-\alpha M_2=\text{tr}(\mb P^\top(\mb X^\top(\mb I-\alpha\mb L)\mb X)\mb P).
\end{align}
    \vspace{-1.5em}

Therefore, our goal is to obtain $\mb P^*$ such that $\mb P^\top(\mb X^\top(\mb I-\alpha\mb L)\mb X)\mb P$ is diagonal and maximize the values on the diagonal. We can obtain the closed-form solution.
Concretely, $\mb P^*$ consists of the top-$d$ eigenvectors of $\mb X^\top(\mb I-\alpha\mb L)\mb X$.
And the new data is computed as $\mb E=(\mb I-\alpha\mb L)^{1/2}\mb X\mb P^*$. The derivation of this process is in Appendix~\ref{app:pca_derive}.
Note that we have to compute $(\mb I-\alpha\mb L)^{1/2}$ to obtain the new data. To avoid this computational load, we approximate it using the Chebyshev expansion as in GCNs~\cite{jiang2019semi}. Appendix~\ref{proof:approx} gives the detailed approximation results.

Finally, the entire process of our Rec-PCA consists of:
\newline
1. Compute $\mb P^*$ as the top-$d$ eigenvectors of $\mb X^\top(\mb I-\alpha\mb L)\mb X$.
\newline
2. Obtain the new data as $\mb E=(\mb I-\alpha\mb L)^{1/2}\mb X\mb P^*$, where $(\mb I-\alpha\mb L)^{1/2}$ is approximated using the Chebyshev expansion.
\newline
3. Use $\mb E$ as the initialization of the item embeddings in backbone sequential recommendation models.

\begin{table*}[!t]
    \caption{Recommendation performance. The best results are in boldface, and the best baselines are underlined. \textit{Improv.b} denotes the relative improvement of \ourmethod over the best baseline. A paired t-test is performed across 5 independent runs to evaluate the $p$-value ($\leq 0.05$ indicates statistical significance).}
    \label{tab:all}
    \vspace{-0.5em}
    \centering
    \resizebox{\linewidth}{!}{
    \begin{tabular}{cc|cccc|cccc|cccc} \toprule
        \multirow{2}{*}{Backbone} & \multirow{2}{*}{Method} & \multicolumn{4}{c|}{Yelp} & \multicolumn{4}{c|}{Sports} & \multicolumn{4}{c}{CDs}\\
         & & H@5 & N@5 & H@10 & N@10 & H@5 & N@5 & H@10 & N@10 & H@5 & N@5 & H@10 & N@10\\
        \midrule
        \multirow{8}{*}{GRU4Rec} & None & 0.0184 & 0.0131& 0.0277 & 0.0176 & 0.0089 & 0.0080 & 0.0092 & 0.0084 & 0.0964 & 0.0789 & 0.1057 & 0.0805\\
         & LLMInit & 0.0197 & 0.0140 & 0.0281 & 0.0182 & 0.0091 & 0.0083 & 0.0095 & 0.0086 & 0.0973 & 0.0792 & 0.1070 & 0.0811\\
         & LLM-ESR & 0.0203 & 0.0156 & 0.0294 & 0.0192 & 0.0094 & 0.0084 & \underline{0.0099} & 0.0089 & \underline{0.0995} & \underline{0.0807} & \underline{0.1104} & \underline{0.0832}\\
         & LLMEmb & 0.0199 & 0.0146 & 0.0289 & 0.0188 & 0.0096 & 0.0085 & 0.0095 & 0.0087 & 0.0977 & 0.0797 & 0.1081 & 0.0815\\
         & LLM2Rec & \underline{0.0208} & \underline{0.0160} & \underline{0.0300} & \underline{0.0197} & \underline{0.0100} & \underline{0.0088} & 0.0099 & \underline{0.0091} & 0.0990 & 0.0801 & 0.1093 & 0.0826\\
         & \ourmethod & \textbf{0.0228} & \textbf{0.0171} & \textbf{0.0324} & \textbf{0.0211} & \textbf{0.0108} & \textbf{0.0096} & \textbf{0.0105} & \textbf{0.0099} & \textbf{0.1047} & \textbf{0.0843} & \textbf{0.1163} & \textbf{0.0868}\\
         \cline{2-14}
         & \textit{Improv.b} & 9.62\% & 6.88\% & 8.00\% & 7.11\% & 8.00\% & 9.10\% & 6.06\% & 8.79\% & 5.23\% & 4.46\% & 5.34\% & 4.33\%\\
         & \textit{p-value} & 3.7e-4 & 5.2e-3 & 6.6e-4 & 8.1e-4 & 7.73e-4 & 6.26e-3 & 9.01e-4 & 5.29e-3 & 4.33e-3 & 5.72e-5 & 1.08e-4 & 9.37e-5\\
        \midrule
        \midrule
         \multirow{8}{*}{Bert4Rec} & None & 0.0231 & 0.0158 & 0.0389 & 0.0202 & 0.0097 & 0.0088 & 0.0109 & 0.0092 & 0.1017 & 0.0819 & 0.1130 & 0.0852\\
         & LLMInit & 0.0247 & 0.0160 & 0.0407 & 0.0209 & 0.0098 & 0.0089 & 0.0107 & 0.0093 & 0.1031 & 0.0829 & 0.1148 & 0.0868\\
         & LLM-ESR & 0.0263 & 0.0166 & 0.0419 & 0.0217 & \underline{0.0107} & \underline{0.0094} & \underline{0.0122} & \underline{0.0097} & 0.1053 & 0.0859 & 0.1166 & 0.0898\\
         & LLMEmb & 0.0271 & 0.0171 & 0.0423 & 0.0221 & 0.0099 & 0.0091 & 0.0112 & 0.0094 & 0.1046 & 0.0846 & 0.1147 & 0.0889\\
         & LLM2Rec & \underline{0.0279} & \underline{0.0180} & \underline{0.0427} & \underline{0.0233} & 0.0104 & 0.0092 & 0.0118 & 0.0095 & \underline{0.1061} & \underline{0.0871} & \underline{0.1187} & \underline{0.0914}\\
         & \ourmethod & \textbf{0.0294} & \textbf{0.0191} & \textbf{0.0453} & \textbf{0.0246} & \textbf{0.0113} & \textbf{0.0098} & \textbf{0.0132} & \textbf{0.0103} & \textbf{0.1107} & \textbf{0.0904} & \textbf{0.1244} & \textbf{0.0946}\\
         \cline{2-14}
         & \textit{Improv.b} & 5.38\% & 6.11\% & 6.09\% & 5.58\% & 5.61\% & 4.26\% & 8.20\% & 6.19\% & 4.34\% & 3.79\% & 4.80\% & 3.50\%\\
         & \textit{p-value} & 2.88e-4 & 5.21e-5 & 9.38e-5 & 4.42e-4 & 7.85e-5 & 9.29e-4 & 5.13e-5 & 4.11e-3 & 7.29e-4 & 8.31e-5 & 7.72e-4 & 6.03e-4\\
         \midrule
         \midrule
         \multirow{8}{*}{SASRec} & None & 0.0247 & 0.0163 & 0.0396 & 0.0210 & 0.0105 & 0.0093 & 0.0118 & 0.0097 & 0.1032 & 0.0838 & 0.1171 & 0.0883\\
         & LLMInit & 0.0243 & 0.0167 & 0.0401 & 0.0213 & 0.0102 & 0.0090 & 0.0111 & 0.0094 & 0.1052 & 0.0843 & 0.1189 & 0.0897\\
         & LLM-ESR & 0.0257 & 0.0179 & 0.0437 & 0.0229 & \underline{0.0111} & \underline{0.0097} & \underline{0.0129} & \underline{0.0101} & 0.1067 & 0.0864 & 0.1210 & 0.0918\\
         & LLMEmb & 0.0268 & 0.0180 & 0.0435 & 0.0236 & 0.0107 & 0.0094 & 0.0120 & 0.0098 & 0.1066 & 0.0866 & 0.1228 & 0.0924\\
         & LLM2Rec & \underline{0.0278} & \underline{0.0186} & \underline{0.0449} & \underline{0.0240} & 0.0109 & 0.0096 & 0.0123 & 0.0099 & \underline{0.1072} & \underline{0.0871} & \underline{0.1243} & \underline{0.0933}\\
         & \ourmethod & \textbf{0.0299} & \textbf{0.0197} & \textbf{0.0472} & \textbf{0.0253} & \textbf{0.0118} & \textbf{0.0101} & \textbf{0.0138} & \textbf{0.0107} & \textbf{0.1118} & \textbf{0.0913} & \textbf{0.1290} & \textbf{0.0969}\\
         \cline{2-14}
         & \textit{Improv.b} & 7.55\% & 5.91\% & 5.12\% & 5.42\% & 6.31\% & 4.12\% & 6.98\% & 5.94\% & 4.29\% & 4.82\% & 3.78\% & 3.86\%\\
         & \textit{p-value} & 2.73e-4 & 1.09e-3 & 5.98e-5 & 6.26e-4 & 7.90e-4 & 2.53e-5 & 5.72e-4 & 1.07e-3 & 5.54e-4 & 6.70e-5 & 4.37e-4 & 2.12e-4\\
        \bottomrule
    \end{tabular}
    }
    \vspace{-1.0em}
\end{table*}

\textbf{Analysis of Superiority.}
As we have discussed, Rec-PCA mitigates semantic–collaboration misaligned angular clustering (i.e., $\kappa(\hat{\mb E}_{\mc U}\hat{\mb E}_{\mc U}^\top)$), thereby reducing the upper bound of $\kappa(\mb H_{h\mc U})$ in Eq.(\ref{eq:kappa_h}). Since we cannot directly calculate $\kappa(\hat{\mb E}_{\mc U}\hat{\mb E}_{\mc U}^\top)$, we now present a series of theoretical analyses to transform measurements of $\kappa(\hat{\mb E}_{\mc U}\hat{\mb E}_{\mc U}^\top)$ into measurements of \emph{Effective Coherence}.
Specifically, we first define the effective coherence as follows.

\begin{definition}[Effective Coherence]
    Given the effective subspace $\mc U$, the effective coherence measures the maximum similarity between items within $\mc U$, which is computed as
    \vspace{-0.5em}
    \begin{align}
    \rho(\hat{\mb E}_{\mc U}):=\max_{i\neq j}\left|\braket{\hat{\bs e}_i,\hat{\bs e}_j}\right|.
\end{align}
\end{definition}
    \vspace{-1.0em}

Then, we have the following theorem that provides the upper bound for $\kappa(\hat{\mb E}_\mc U\hat{\mb E}_\mc U^\top)$ using the effective coherence.

\begin{theorem}\label{theorem:norm}
    When $(m-1)\rho(\hat{\mb E}_{\mc U})<1$, we have
    \vspace{-0.5em}
    \begin{align}
        \kappa(\hat{\mb E}_\mc U\hat{\mb E}_\mc U^\top)\le \frac{1+(m-1)\rho(\hat{\mb E}_{\mc U})}{1-(m-1)\rho(\hat{\mb E}_{\mc U})}.
    \end{align}
    \vspace{-1.5em}
\end{theorem}

The proof is presented in Appendix~\ref{proof:norm}. By analyzing the monotonicity of the terms on the right-hand side of the inequality, we can finally derive the following proposition.
\begin{proposition}\label{theorem:rho_gcn}
    When $(m-1)\rho(\hat{\mb E}_{\mc U})<1$, $\kappa(\hat{\mb E}_\mc U\hat{\mb E}_\mc U^\top)$ increases with $\rho(\hat{\mb E}_{\mc U})$.
\end{proposition}
    \vspace{-0.5em}
The proof is provided in Appendix~\ref{proof:rho_gcn}.
That is, if our measurement indicate that Rec-PCA reduces $\rho(\hat{\mb E}_{\mc U})$, or equivalently, increases the distinguishability between positive samples and hard negative samples, we can expect it to reduce $\kappa(\hat{\mb E}_\mc U\hat{\mb E}_\mc U^\top)$, thereby lowering the upper bound of $\kappa(\mb H_{h\mc U})$.

In a spectral filtering view, minimizing the total variation of representations on the graph attenuates task-mismatched (or equivalently, graph-inconsistent) components that dominate recommendation-misaligned similarity.
We empirically validate this in Figure~\ref{fig:rho_main}. We observe that Rec-PCA consistently reduces the effective coherence. 
Finally, as the most intuitive demonstration, Figure~\ref{fig:loss_main} shows the loss curve throughout training with and without Rec-PCA. We observe that Rec-PCA significantly facilitates the decline of the training loss. Throughout the training process, its loss remains significantly lower than other comparison methods. This effectively validates that our method facilitates training.

\section{Experiments}

We conduct experiments on three public datasets using three backbones. \textbf{Implementation details} are provided in Appendix~\ref{app:implement}. The overall performance comparison is presented in Section~\ref{sec:result}. Then, Section~\ref{sec:compat} demonstrates our compatibility with existing methods by integrating them. The ablation study is provided in Section~\ref{sec:abl}. Some key \textbf{hyperparameter analysis} is provided in Appendix~\ref{sec:hyper}.

\subsection{Experimental Settings}\label{subsec:exp-setting}

\noindent\textbf{Datasets.}
We adopt three datasets for verification, i.e., \textbf{Yelp}, \textbf{Amazon Sports}, and \textbf{Amazon CDs}. Statistics and methods of training-test splitting are in Appendix~\ref{sec:dataset}.

\noindent\textbf{Evaluation Protocols.}
To evaluate the performance, we employ Hit Rate (H@$N$) and Normalized Discounted Cumulative Gain (N@$N$) and report the results for $N \in \{5,10\}$.

\noindent\textbf{Baselines.}
We compare our method with four up-to-date LLM-enhanced sequential recommendation methods, including LLMInit~\cite{zhang2025llminit}, LLMEmb~\cite{liu2025llmemb}, LLM-ESR~\cite{liu2024llm}, and LLM2Rec~\cite{he2025llm2rec}. Details about baselines are in Appendix~\ref{app:baseline}.

\noindent\textbf{Backbones.}
Following recent work~\cite{liu2025llmemb,liu2024llm,he2025llm2rec}, we test our method on three well-known models, i.e., GRU4Rec~\cite{hidasi2015session}, Bert4Rec~\cite{sun2019bert4rec}, and SASRec~\cite{kang2018self}. Details about backbones are in Appendix~\ref{app:backbone}.

\subsection{Performance Comparison}\label{sec:result}
The performance of all methods is provided in Table~\ref{tab:all}. From the results, we observe that:
\newline
\textbf{(\romannumeral1)} \ourmethod outperforms all baselines significantly on all datasets and achieves remarkable improvements over the best baseline for all backbones. We've seen significant growth on Yelp and Amazon Sports, with increases exceeding 5\% across the board.
\newline 
\textbf{(\romannumeral2)} Our approach demonstrates excellent performance across all backbones while preserving their relative strengths. This indicates that our method exhibits strong generalization across backbones and has broad applicability.

\begin{table}[htbp]\renewcommand{\arraystretch}{1.2}
    \caption{Compatibility with existing methods.}
    \label{tab:compat}
    \vspace{-0.5em}
    \centering
    \resizebox{\linewidth}{!}{
    \begin{tabular}{cc|cc|cc} \toprule
        \multirow{2}{*}{Backbone} & \multirow{2}{*}{Method} & \multicolumn{2}{c|}{Yelp} & \multicolumn{2}{c}{Sports}\\
         & & H@10 & N@10 & H@10 & N@10\\
        \midrule
        \multirow{6}{*}{GRU4Rec} & \multicolumn{1}{l|}{LLMemb} & 0.0289 & 0.0188 & 0.0095 & 0.0087\\
         & \multicolumn{1}{r|}{+ours} & 0.0321 & 0.0206 & 0.0103 & 0.0094\\
         & \textit{Improv.b} & 11.07\% & 9.57\% & 8.42\% & 8.05\%\\
         \cline{2-6}
         & \multicolumn{1}{l|}{LLM2Rec} & 0.0300 & 0.0197 & 0.0099 & 0.0091\\
         & \multicolumn{1}{r|}{+ours} & 0.0327 & 0.0213 & 0.0109 & 0.0100\\
         & \textit{Improv.b} & 9.00\% & 8.12\% & 10.10\% & 9.89\%\\
        \midrule
        \midrule
        \multirow{6}{*}{Bert4Rec} & \multicolumn{1}{l|}{LLMemb} & 0.0423 & 0.0221 & 0.0112 & 0.0094\\
         & \multicolumn{1}{r|}{+ours} & 0.0462 & 0.0241 & 0.0124 & 0.0105\\
         & \textit{Improv.b} & 9.22\% & 9.04\% & 10.71\% & 11.70\%\\
         \cline{2-6}
         & \multicolumn{1}{l|}{LLM2Rec} & 0.0427 & 0.0233 & 0.0118 & 0.0095\\
         & \multicolumn{1}{r|}{+ours} & 0.0474 & 0.0258 & 0.0134 & 0.0104\\
         & \textit{Improv.b} & 11.01\% & 10.73\% & 13.56\% & 9.47\%\\
        \midrule
        \midrule
        \multirow{6}{*}{SASRec} & \multicolumn{1}{l|}{LLMemb} & 0.0435 & 0.0236 & 0.0120 & 0.0098\\
         & \multicolumn{1}{r|}{+ours} & 0.0480 & 0.0260 & 0.0129 & 0.0107\\
         & \textit{Improv.b} & 10.34\% & 10.17\% & 7.50\% & 9.18\%\\
         \cline{2-6}
         & \multicolumn{1}{l|}{LLM2Rec} & 0.0449 & 0.0240 & 0.0123 & 0.0099\\
         & \multicolumn{1}{r|}{+ours} & 0.0503 & 0.0269 & 0.0141 & 0.0111\\
         & \textit{Improv.b} & 12.03\% & 12.08\% & 14.63\% & 12.12\%\\
        \bottomrule
    \end{tabular}
    }
    \vspace{-1.0em}
\end{table}

\subsection{Compatibility with Existing Methods}\label{sec:compat}

Our approach can be easily combined with other LLM-enhanced methods. Specifically, all LLM-enhanced methods comprise two operations: 1) integrating LLM representations into backbone recommendation models, and 2) subsequently retraining the backbone models. When integrating textual representations into backbone recommendation models, if PCA is originally used to reduce LLM representations, we directly replace it with Rec-PCA. Otherwise, we insert Rec-PCA at this step, using the reduced representations as new LLM representations. Moreover, when re-training the backbone recommendation models, it is unavoidable to compute item logits. Here, we use normalized item embeddings for this computation.
Following the above steps, we integrate our method into two state-of-the-art LLM-enhanced recommendation models, i.e., LLMEmb~\cite{liu2025llmemb} and LLM2Rec~\cite{he2025llm2rec}. 

We report the results in Table~\ref{tab:compat}. After integrating our method (denoted as ``+ours''), we significantly outperform the original methods across all datasets and backbones, demonstrating the effectiveness of our augmentation on existing models. Moreover, the idea can be applied to various methods with different performances, while maintaining the relative superiority of the original methods.

\begin{figure}
\centering
  \includegraphics[width=\linewidth]{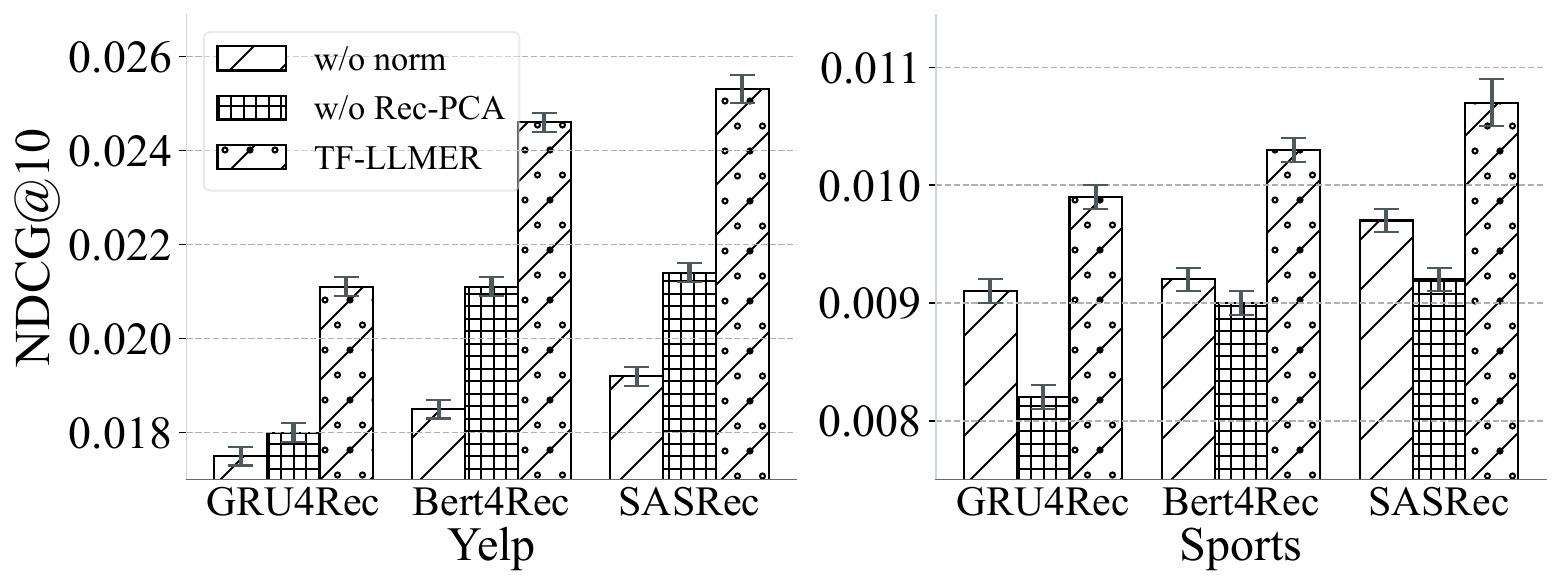}
    \vspace{-1.7em}
  \caption{Ablation study of our method, including three variants.}
  \label{fig:abl_main}
    \vspace{-1.7em}
\end{figure}

\subsection{Ablation Study}\label{sec:abl}

To validate the effectiveness of our key components, we design two variants: 1) \textbf{\ourmethod w/o norm} retains Rec-PCA but uses unnormalized item embeddings to compute the logits; 2) \textbf{\ourmethod w/o Rec-PCA} retains embedding normalization but replaces Rec-PCA with vanilla PCA. 

Figure~\ref{fig:abl_main} shows the results of these two variants, together with the full version of \ourmethod. The full results are given in Appendix~\ref{app:abl}.
We find that all variants are inferior to the original \ourmethod, which demonstrates the effectiveness of all key components. Moreover, we observe that item normalization has a large effect on Yelp, while Rec-PCA has a greater impact on Amazon Sports. We believe this relates to the characteristics of the item representations generated by the LLM across different datasets. On Yelp and Amazon CDs, the LLM-generated item representations exhibit larger magnitude differences, leading to more significant training challenges.
On Amazon Sports, increasing the distinguishability between effective items significantly improves training effectiveness, thereby amplifying Rec-PCA's impact on recommendation performance.

\section{Conclusion}

This work revisits LLM-enhanced recommendation from a representation-level optimization perspective. Motivated by the observation that LLM-derived item embeddings exhibit substantial norm disparities, thereby inducing severe training ill-conditioning, we highlight the importance of normalizing embeddings when computing logits, thereby yielding provable conditioning bounds.
Furthermore, we address the misalignment between semantic signals and collaborative structure by proposing Rec-PCA. It injects sequence-derived co-occurrence relations into dimensionality reduction via a controllable total-variation term. Together, normalization and Rec-PCA form a lightweight framework (\ourmethod).



\bibliography{refs}
\bibliographystyle{icml2026}

\newpage
\appendix
\onecolumn
\section{Proofs}

\subsection{Derivation of Hessian Matrixes}\label{proof:h}

The detailed formula of $\mb H_s$ is
\begin{align}\label{eq:def_hs}
    \mb H_s:=\nabla_{\bs s}^2(-\log p_y)=\text{Diag}(\bs p)-\bs p\bs p^\top,
\end{align}
where $\bs p=\softmax(\bs s)$ and $\text{Diag}(\bs p)$ denotes the diagonal matrix whose values are $p_i$.

\begin{proof}
Recall that $\ell(\bs h)=-\log p_y$, where
\begin{align}
    \bs p=\softmax(\bs s)\qquad\text{and}\qquad p_i=\frac{e^{s_i}}{\sum_{j=1}^{|\mc I|}e^{s_j}}.
\end{align}

It can be derived that
\begin{align}
    \nabla_{\bs s}\ell=\frac{\partial}{\partial\bs s}(-\log p_y)=\bs p-\bs\delta_y,
\end{align}
where $\bs\delta_y$ is a one-hot vector centered on $y$.

Then, using the chain rule, we have the gradient w.r.t. $\bs h$ given by
\begin{align}
    \nabla_{\bs h}\ell=\frac{\partial}{\partial\bs h}(-\log p_y)=\mb E^\top\nabla_{\bs s}\ell=\mb E^\top(\bs p-\bs\delta_y)
\end{align}

Having the gradients, we obtain the Hessian w.r.t. $\bs s$ as follows:
\begin{align}
    \mb H_{\bs s}:=\nabla_{\bs s}^2\ell=\frac{\partial}{\partial\bs s}(\bs p-\bs\delta_y)=\text{Diag}(\bs p)-\bs p\bs p^\top,
\end{align}
where $\text{Diag}(\bs p)$ denotes the diagonal matrix whose values are $p_i$.

Finally, since $\bs s=\mb E\bs h$, we obtain the Hessian w.r.t. $\bs h$ using the chain rules as follows:
\begin{align}
    \mb H_{\bs h}:=\nabla_{\bs h}^2\ell=\nabla_{\bs h}(\mb E^\top\nabla_{\bs s}\ell)=\mb E^\top\mb H_{\bs s}\mb E.
\end{align}

\end{proof}

\subsection{Proof of Theorem~\ref{theorem:norm_disparity}}\label{proof:norm_disparity}

We first introduce the following lemma.
\begin{lemma}\label{lemma:kappa}
    In the effective subspace $\mc U$, we obtain
    \begin{align}
        \kappa(\mb H_{h\mc U})\le \frac{\beta}{\alpha}\cdot \kappa\left(\mb E_{\mc U}^\top\mb E_{\mc U}\right),
    \end{align}
    where $\mb E_{\mc U}:=\mc P_{\mc U}\mb E$ denotes the projection of $\text{ }\mb E$ onto the subspace $\mc U$.
\end{lemma}
\begin{proof}
    Assumption~\ref{ass:effective} implies that
    \begin{align}
        \alpha\mb I_{\mc U}\preceq\mb H_{s\mc U}\preceq\beta\mb I_{\mc U},
    \end{align}
    where $\mb I_{\mc U}:=\mc P_{\mc U}\mb I\mc P_{\mc U}^\top$ and $\mb H_{s\mc U}:=\mc P_{\mc U}\mb H_s\mc P_{\mc U}^\top$ denote the projection of $\mb I$ and $\mb H_s$ onto the subspace $\mc U$, respectively.
    
    Then, combining with the definition of $\mb H_{h\mc U}$, i.e., $\mb H_{h\mc U}:=\mb E_{\mc U}^\top\mb H_{s\mc U}\mb E_{\mc U}$, yields that, in subspace $\mc U$,
    \begin{align}
        \alpha\mb E_{\mc U}^\top\mb E_{\mc U}\preceq\mb H_{h\mc U}\preceq\beta\mb E_{\mc U}^\top\mb E_{\mc U}.
    \end{align}
    
    Taking the maximal and minimal eigenvalues yields that, in the subspace $\mc U$,
    \begin{align}
        \lambda_{\max}(\mb H_{h\mc U})\le\beta\lambda_{\max}(\mb E_{\mc U}^\top\mb E_{\mc U})\quad\text{and}\quad \lambda_{\min}(\mb H_{h\mc U})\ge\alpha\lambda_{\min}(\mb E_{\mc U}^\top\mb E_{\mc U}).
    \end{align}
    Dividing the two inequalities gives the lemma.
\end{proof}

Then, to obtain the upper bound of $\kappa(\mb H_{h\mc U})$, we first introduce the following lemma that gives the upper bound of $\kappa(\mb E_{\mc U}^\top\mb E_{\mc U})$.

\begin{lemma}\label{lemma:kappa_gram}
    Suppose that $\mb E_{\mc U}^\top\mb E_{\mc U}$ and $\hat{\mb E}_{\mc U}^\top\hat{\mb E}_{\mc U}$ are positive definit. The norm disparity of items quadratically deteriorates the upper bound of $\kappa(\mb E_{\mc U}^\top\mb E_{\mc U})$. Formally,
    \begin{align}
        \kappa(\mb E_{\mc U}^\top\mb E_{\mc U})\le \left(\frac{r_{\max}}{r_{\min}}\right)^2\cdot\kappa(\hat{\mb E}_{\mc U}^\top\hat{\mb E}_{\mc U}).
    \end{align}
    Moreover, this upper bound is achievable even when the normalized item embedding is well-conditioned.
\end{lemma}
\begin{proof}
    Let 
    \begin{align}
        \mb B := \mb E_{\mc U}^\top\mb E_{\mc U}\qquad\text{and}\qquad\mb G:=\hat{\mb E}_{\mc U}^\top\hat{\mb E}_{\mc U}.
    \end{align}
    We separate the magnitude and angle of the item embeddings by letting $\bs e_i=r_i\hat{\bs e}_i$, where $r_i\in \R$ and $\hat{\bs e}_i\in\R^d$ denotes the norm of $\bs e_i$ and the normalized embedding, respectively. Then, we have
    \begin{align}
        \mb B=\sum_{i\in\mc U}r_i^2\hat{\bs e}_i\hat{\bs e}_i^\top\qquad\text{and}\qquad\mb G=\sum_{i\in\mc U}\hat{\bs e}_i\hat{\bs e}_i^\top.
    \end{align}
    Note that, for all $i\in\mc U$, we have
    \begin{align}
        r_{\min}^2\le r_i^2\le r_{\max}^2.
    \end{align}
    Since $\hat{\bs e}_i\hat{\bs e}_i^\top$ is positive semi-definite and multiplying a positive semi-definite matrix maintains the Loewner order, we obtain
    \begin{align}
        r_{\min}^2\hat{\bs e}_i\hat{\bs e}_i^\top\preceq r_i^2\hat{\bs e}_i\hat{\bs e}_i^\top\le r_{\max}^2\hat{\bs e}_i\hat{\bs e}_i^\top,\quad\forall i\in\mc U.
    \end{align}
    Consequently,
    \begin{align}
        r_{\min}^2\sum_{i\in\mc U}\hat{\bs e}_i\hat{\bs e}_i^\top\preceq \sum_{i\in\mc U}r_i^2\hat{\bs e}_i\hat{\bs e}_i^\top\le r_{\max}^2\sum_{i\in\mc U}\hat{\bs e}_i\hat{\bs e}_i^\top,\quad\forall i\in\mc U.
    \end{align}
    Namely,
    \begin{align}
        r_{\min}^2\mb G\preceq\mb B\preceq r_{\max}^2\mb G.
    \end{align}
    Using Courant-Fischer Minimax Theorem~\cite{ikebe1987monotonicity}, we see
    \begin{align}
        \lambda_{\max}(\mb B)\le \lambda_{\max}(r_{\max}^2\mb G)=r_{\max}^2\lambda_{\max}(\mb G)\\
        \lambda_{\min}(\mb B)\ge \lambda_{\min}(r_{\min}^2\mb G)=r_{\min}^2\lambda_{\min}(\mb G).
    \end{align}
    Since $\mb B$ and $\mb G$ are positive definit, we have
    \begin{align}
        \kappa(\mb B)=\frac{\lambda_{\max}(\mb B)}{\lambda_{\min}(\mb B)}\le \left(\frac{r_{\max}}{r_{\min}}\right)^2\cdot\kappa(\mb G).
    \end{align}

\end{proof}

\begin{lemma}\label{lemma:eigen}
    Let $\mb A=\mb X\mb X^\top$ and $\mb B=\mb X^\top\mb X$, where $\mb X\in\R^{m\times n}$. Then, $\mb A$ and $\mb B$ share the same eigenvalues.
\end{lemma}
\begin{proof}
It is obvious that the $\mb A$ and $\mb B$ have the same rank. Let the rank of $\mb A$ and $\mb B$ be $d$. Suppose that $\{\lambda_i\}_{i=1}^d$ be the non-zero eigenvalues of $\mb A$, and $\{\bs u_i\}_{i=1}^d$ be the corresponding eigenvectors. It gives
\begin{align}
    \mb A\bs u_i=\lambda_i\bs u_i,\quad\forall i=1,\cdots,d.
\end{align}
Then,
\begin{align}
    \mb B\cdot(\mb X^\top\bs u_i)=\mb X^\top\mb X \mb X^\top\bs u_i=\mb X^\top\mb A\bs u_i=\lambda_i\cdot(\mb X^\top\bs u_i),\quad\forall i=1,\cdots,d.
\end{align}
Finally, we can conclude that $\{\lambda_i\}_{i=1}^d$ is also the eigenvalues of $\mb B$, and the corresponding eigenvectors are $\{\mb X^\top\bs u_i\}_{i=1}^d$.
\end{proof}

\textbf{Proof of Theorem~\ref{theorem:norm_disparity}.}
Firstly, by combining Lemma~\ref{lemma:kappa} with Lemma~\ref{lemma:kappa_gram}, we obtain that
\begin{align}
    \kappa(\mb H_{h\mc U})\le \frac{\beta}{\alpha}\cdot\left(\frac{r_{\max}}{r_{\min}}\right)^2\cdot\kappa(\hat{\mb E}_{\mc U}^\top\hat{\mb E}_{\mc U}).\label{eq:lemma}
\end{align}
Then, Lemma~\ref{lemma:eigen} gives that
\begin{align}
    \kappa(\hat{\mb E}_{\mc U}^\top\hat{\mb E}_{\mc U})=\frac{\lambda_{\max}(\hat{\mb E}_{\mc U}^\top\hat{\mb E}_{\mc U})}{\lambda_{\min}(\hat{\mb E}_{\mc U}^\top\hat{\mb E}_{\mc U})}=\frac{\lambda_{\max}(\hat{\mb E}_{\mc U}\hat{\mb E}_{\mc U}^\top)}{\lambda_{\min}(\hat{\mb E}_{\mc U}\hat{\mb E}_{\mc U}^\top)}=\kappa(\hat{\mb E}_{\mc U}\hat{\mb E}_{\mc U}^\top).
\end{align}
Substituting it into inequality~(\ref{eq:lemma}) completes the proof.

\subsection{Proof of Theorem~\ref{theorem:unnorm}}\label{proof:unnorm}

\begin{proof}
    Recall that $\mb H_{s\mc U}=\mc P_{\mc U}\mb H_s\mc P_{\mc U}^\top$. Under Assumption~\ref{ass:effective}, $\mb H_{s\mc U}$ is positive definite and can be decomposed as
    \begin{align}
        \mb H_{s\mc U}=\mb S^\top\mb S,
    \end{align}
    where $\mb S:=(\mb H_{s\mc U})^{1/2}$.
    Then, 
    \begin{align}
        \mb H_{h\mc U}=\mb E_{\mc U}^\top\mb H_{s\mc U}\mb E_{\mc U}=(\mb S\mb E_{\mc U})^\top(\mb S\mb E_{\mc U}).
    \end{align}
    Denote by $\sigma_{\min}(\cdot)$ and $\sigma_{\max}(\cdot)$ the minimum and maximum singular values, respectively.
    Therefore, we have
    \begin{align}
        \kappa(\mb H_{h\mc U})=\frac{\lambda_{\max}((\mb S\mb E_{\mc U})^\top(\mb S\mb E_{\mc U}))}{\lambda_{\min}((\mb S\mb E_{\mc U})^\top(\mb S\mb E_{\mc U}))}=\left(\frac{\sigma_{\max}(\mb S\mb E_{\mc U})}{\sigma_{\min}(\mb S\mb E_{\mc U})}\right)^2.
    \end{align}
    By the definition of the Rayleigh quotient, we know that for any unit vector $\bs x$, we have
    \begin{align}
        \lambda_{\min}(\mb H_{s\mc U})\cdot\|\mb E_{\mc U}\bs x\|^2\le (\mb E_{\mc U}\bs x)^\top(\mb S^\top\mb S)(\mb E_{\mc U}\bs x)=\|\mb S\mb E_{\mc U}\bs x\|^2.
    \end{align}
    Equally,
    \begin{align}
        \|\mb S\mb E_{\mc U}\bs x\|\ge \sqrt{\lambda_{\min}(\mb H_{s\mc U})}\|\mb E_{\mc U}\bs x\|.
    \end{align}
    Under Assumption~\ref{ass:effective}, we have $\lambda_{\min}(\mb H_{s\mc U})\ge \alpha$, thus,
    \begin{align}
        \|\mb S\mb E_{\mc U}\bs x\|\ge \sqrt{\alpha}\|\mb E_{\mc U}\bs x\|.
    \end{align}
    The variational characterization of singular values~\cite{horn2012matrix} gives that
    \begin{align}
        \sigma_{\max}(\mb S\mb E_{\mc U})=\max_{\|\bs x\|=1}\|\mb S\mb E_{\mc U}\bs x\|\quad \text{and}\quad\sigma_{\max}(\mb E_{\mc U})=\max_{\|\bs x\|=1}\|\mb E_{\mc U}\bs x\|.
    \end{align}
    Therefore, we get
    \begin{align}
        \sigma_{\max}(\mb S\mb E_{\mc U})\ge \sqrt{\alpha}\cdot\sigma_{\max}(\mb E_{\mc U}).
    \end{align}
    Simillary, for any unit vector $\bs x$, we have
    \begin{align}
        \lambda_{\max}(\mb H_{s\mc U})\cdot\|\mb E_{\mc U}\bs x\|^2\ge (\mb E_{\mc U}\bs x)^\top(\mb S^\top\mb S)(\mb E_{\mc U}\bs x)=\|\mb S\mb E_{\mc U}\bs x\|^2,
    \end{align}
    and
    \begin{align}
        \|\mb S\mb E_{\mc U}\bs x\|\le \sqrt{\lambda_{\max}(\mb H_{s\mc U})}\|\mb E_{\mc U}\bs x\|.
    \end{align}
    Therefore, for any unit vector $\bs x$,
    \begin{align}
        \|\mb S\mb E_{\mc U}\bs x\|\le \sqrt{\beta}\|\mb E_{\mc U}\bs x\|.
    \end{align}
    Then, the variational characterization of singular values gives that
    \begin{align}
        \sigma_{\min}(\mb S\mb E_{\mc U})\le \sqrt{\beta}\cdot\sigma_{\min}(\mb E_{\mc U}).
    \end{align}
    Finally, we have
    \begin{align}
        \kappa(\mb H_{h\mc U})=\left(\frac{\sigma_{\max}(\mb S\mb E_{\mc U})}{\sigma_{\min}(\mb S\mb E_{\mc U})}\right)^2\ge \frac{\alpha}{\beta}\left(\frac{\sigma_{\max}(\mb E_{\mc U})}{\sigma_{\min}(\mb E_{\mc U})}\right)^2=\frac{\alpha}{\beta}\cdot\kappa(\mb E_{\mc U}\mb E_{\mc U}^\top).
    \end{align}
\end{proof}

\subsection{Proof of Proposition~\ref{prop:unnorm}}\label{proof:prop_unnorm}
We first introduce the following lemma that gives the unbounded nature of $\kappa(\mb E_{\mc U}\mb E_{\mc U}^\top)$.
\begin{lemma}\label{lemma:unnorm}
    Given that the norms of item embeddings are unconstrained, then even if $\kappa(\hat{\mb E}_{\mc U}\hat{\mb E}_{\mc U}^\top)$ is a finite, $\kappa(\mb E_{\mc U}\mb E_{\mc U}^\top)$ can be arbitrarily large.
\end{lemma}
\begin{proof}
    Here we provide a constructive proof.
    
    Take $d=2$ and choose two effective rows in $\mc U$:
    \begin{align}
        \bs e_1=\begin{pmatrix}
            R & 0
        \end{pmatrix}^\top,\quad \bs e_2=\begin{pmatrix}
            0 & 1
        \end{pmatrix}^\top,
    \end{align}
    with $R\to \infty$. Then,
    \begin{align}
        \mb E_{\mc U}\mb E_{\mc U}^\top=\begin{pmatrix}
            R^2 & 0\\
            0 & 1
        \end{pmatrix},
    \end{align}
    so $\kappa(\mb E_{\mc U}\mb E_{\mc U}^\top)=R^2\to \infty$. And, for any $M>0$, we can always construct $\mb E$ such that $\kappa(\mb E_{\mc U}\mb E_{\mc U}^\top)\ge M$ by setting $R$ large enough.
    
    For other $d$, we can similarly construct the embeddings that $\kappa(\mb E_{\mc U}\mb E_{\mc U}^\top)\ge M$ for any $M>0$.
\end{proof}

Then, combining with Theorem~\ref{theorem:unnorm}, we conclude that $\kappa(\mb H_{h\mc U})$ can also be arbitrarily large.

\subsection{Proof of Theorem~\ref{theorem:norm}}\label{proof:norm}

To prove Theorem~\ref{theorem:norm}, we first present the following lemma.

\begin{lemma}[Gershgorin bound for effective cosine similarity matrix]
    Let $m=\text{dim}(\mc U)$. The smallest and largest eigenvalues of the cosine similarity matrix $\hat{\mb E}_\mc U\hat{\mb E}_\mc U^\top$ can be bounded as:
    \begin{align}
        \lambda_{\min}(\hat{\mb E}_\mc U\hat{\mb E}_\mc U^\top)\ge 1-(m-1)\rho(\hat{\mb E}_{\mc U}), \quad \lambda_{\max}(\hat{\mb E}_\mc U\hat{\mb E}_\mc U^\top)\le 1+(m-1)\rho(\hat{\mb E}_{\mc U}).
    \end{align}
\end{lemma}
\begin{proof}
    Let $\mb C:=\hat{\mb E}_{\mc U}\hat{\mb E}_{\mc U}^\top$.
    
    Since $\mb C_{ii}=\|\hat{\bs e}_{i}\|^2=1$, where $\hat{\bs e}_{i}$ denotes the $i$-th row of $\hat{\mb E}_{\mc U}$, we have $\mb C_{ii}=1$.

    By the definition of $\rho(\hat{\mb E}_{\mc U})$, we have $|\braket{\hat{\bs e}_i,\hat{\bs e}_j}|\le \rho(\hat{\mb E}_{\mc U}),\forall i\neq j\in \mc U$. Or equally, $|\mb C_{ij}|\le \rho(\hat{\mb E}_{\mc U}),\forall i\neq j$.

    For each row $i$, the Gershgorin radius~\cite{varga1962iterative} satisfies
    \begin{align}
        R_i=\sum_{j\neq i}|\mb C_{ij}|\le (m-1)\rho(\hat{\mb E}_{\mc U}).
    \end{align}

    By the Gershgorin circle theorem~\cite{varga1962iterative}, all eigenvalues of $\mb C$ lie in $[1-(m-1)\rho(\hat{\mb E}_{\mc U}), 1+(m-1)\rho(\hat{\mb E}_{\mc U})]$.
    
    By Lemma~\ref{lemma:eigen}, we have that $\hat{\mb E}_{\mc U}^\top\hat{\mb E}_{\mc U}$ and $\hat{\mb E}_{\mc U}\hat{\mb E}_{\mc U}^\top$ share the same eigenvalues. Thus, all eigenvalues of $\hat{\mb E}_{\mc U}^\top\hat{\mb E}_{\mc U}$ also lie in $[1-(m-1)\rho(\hat{\mb E}_{\mc U}), 1+(m-1)\rho(\hat{\mb E}_{\mc U})]$.
\end{proof}

\textbf{Proof of Theorem~\ref{theorem:norm}.}
When $1-(m-1)\rho(\hat{\mb E}_{\mc U})>0$, we have $\lambda_{\min}>0$. The condition number of $\hat{\mb E}_{\mc U}\hat{\mb E}_{\mc U}^\top$ satisfies
\begin{align}
    \kappa(\hat{\mb E}_{\mc U}\hat{\mb E}_{\mc U}^\top)=\kappa(\hat{\mb E}_{\mc U}^\top\hat{\mb E}_{\mc U})=\frac{\lambda_{\max}(\hat{\mb E}_{\mc U}^\top\hat{\mb E}_{\mc U})}{\lambda_{\min}(\hat{\mb E}_{\mc U}^\top\hat{\mb E}_{\mc U})}\le \frac{1+(m-1)\rho(\hat{\mb E}_{\mc U})}{1-(m-1)\rho(\hat{\mb E}_{\mc U})}.
\end{align}

Combining with Lemma~\ref{lemma:kappa}, we conclude that 
\begin{align}
    \kappa(\mb H_{h\mc U})\le \frac{\beta}{\alpha}\cdot\frac{1+(m-1)\rho(\hat{\mb E}_{\mc U})}{1-(m-1)\rho(\hat{\mb E}_{\mc U})}.
\end{align}

\subsection{Proof of Proposition~\ref{theorem:rho_gcn}}\label{proof:rho_gcn}

\begin{proof}
    Let 
    \begin{align}
        f(\rho)=\frac{1+(m-1)\rho}{1-(m-1)\rho}.
    \end{align}
    The derivative
    \begin{align}
        f'(\rho)=\frac{2(m-1)}{(1-(m-1)\rho)^2}>0 \quad\text{when}\quad (m-1)\rho<1,
    \end{align}
    indicating that $f(\rho)$ increases with $\rho$.
\end{proof}

\subsection{Derivation of Rec-PCA}\label{app:pca_derive}

Recall that our objective is
\begin{align}
    \mb P^*=\arg\max_{\mb P}\text{tr}(\mb P^\top(\mb X^\top(\mb I-\alpha\mb L)\mb X)\mb P)\label{eq:P*}
\end{align}
and we further require that $\mb P^\top(\mb X^\top(\mb I-\alpha\mb L)\mb X)\mb P$ is diagonal so that the resulting dimensions have zero covariance. 
Let
\begin{align}
\mb S \ :=\ \mb X^\top(\mb I-\alpha\mb L)\mb X \in \mathbb{R}^{d_{\mathrm{LLM}}\times d_{\mathrm{LLM}}}.
\end{align}
Since $\mb L$ is symmetric and $\mb X$ is real, $\mb S$ is symmetric. Moreover, the eigenvalues of $\mb L$ satisfy $\lambda(\mb L)\subset[0,2]$. Hence, for $\alpha\in[0,0.5]$, we have $\mb I-\alpha \mb L\succeq 0$ and the principal square root $(\mb I-\alpha \mb L)^{1/2}$ is well-defined.

\paragraph{Step 1: Diagonalization implies choosing an eigen-basis of $\mb S$.}
Let the eigen-decomposition of $\mb S$ be
\begin{align}
\mb S = \mb Q\mb \Lambda \mb Q^\top,\quad 
\mb \Lambda=\mathrm{Diag}(\lambda_1(\mb S),\dots,\lambda_{d_{\mathrm{LLM}}}(\mb S)),\quad
\lambda_1(\mb S)\ge \lambda_2(\mb S)\ge \cdots \ge \lambda_{d_{\mathrm{LLM}}}(\mb S).
\end{align}
If $\mb P$ consists of $d$ orthonormal eigenvectors of $\mb S$, i.e., $\mb P = \mb Q_{(:,\mathcal{A})}$ for some index set $\mathcal{A}$ with $|\mathcal{A}|=d$, then
\begin{align}
\mb P^\top \mb S \mb P 
= \mb Q_{(:,\mathcal{A})}^\top (\mb Q\mb \Lambda \mb Q^\top) \mb Q_{(:,\mathcal{A})}
= \mb \Lambda_{(\mathcal{A},\mathcal{A})},
\end{align}
which is diagonal. Therefore, to satisfy the diagonal-covariance requirement, it is natural (and sufficient) to choose $\mb P$ from an eigen-basis of $\mb S$.

\paragraph{Step 2: Maximizing the diagonal entries selects the top-$d$ eigenvectors.}
For any feasible $\mb P$ with $\mb P^\top \mb P = \mb I$, we can write
\begin{equation}
\mathrm{tr}(\mb P^\top \mb S \mb P)
= \sum_{k=1}^d \mb P_k^\top \mb S \mb P_k,
\tag{A.67}
\end{equation}
where $\mb P_k$ is the $k$-th column of $\mb P$ and each term is a Rayleigh quotient. By the Ky Fan maximum principle\cite{marcus1957maximum}, we have
\begin{align}
\max_{\mb P^\top \mb P=\mb I}\ \mathrm{tr}(\mb P^\top \mb S \mb P)
= \sum_{k=1}^d \lambda_k(\mb S),
\end{align}
and the maximizer is attained by taking $\mb P^*$ as the top-$d$ eigenvectors of $\mb S$, i.e.,
\begin{align}
\mb P^* = \mb Q_{(:,1:d)}.
\end{align}
Consequently,
\begin{align}
(\mb P^*)^\top \mb S \mb P^* = \mathrm{Diag}(\lambda_1(\mb S),\dots,\lambda_d(\mb S)),
\end{align}
which is diagonal and has the largest possible diagonal sum.

\paragraph{Step 3: The transformed data $\mb E$ and its covariance.}
Define
\begin{align}
\widetilde{\mb X} := (\mb I-\alpha \mb L)^{1/2} \mb X.
\end{align}
Then $\mb S = \mb X^\top (\mb I-\alpha \mb L)\mb X = \widetilde{\mb X}^\top \widetilde{\mb X}$. 
Following the standard PCA form, we compute the reduced representations as
\begin{align}
\mb E \ :=\ \widetilde{\mb X}\mb P^* \ =\ (\mb I-\alpha L\mb )^{1/2} \mb X \mb P^*.
\end{align}
Its covariance satisfies
\begin{align}
\mb E^\top \mb E
= (\mb P^*)^\top \widetilde{\mb X}^\top \widetilde{\mb X} \mb P^*
= (\mb P^*)^\top \mb X^\top (\mb I-\alpha \mb L)\mb X \mb P^*
= \mathrm{Diag}(\lambda_1(\mb S),\dots,\lambda_d\mb (\mb S)),
\end{align}
which verifies: (i) zero off-diagonal covariance, and (ii) the diagonal entries are exactly the top-$d$ variances, hence optimal under Eq.(\ref{eq:P*}).
This completes the derivation of the closed-form solution stated in the main text.

\subsection{Approximation Results of $(\mb I-\alpha\mb L)^{1/2}$}\label{proof:approx}

Consistent with the GCN observations, we find that low-order approximations already perform well. Therefore, we present the first-order and second-order approximation results here.

\textbf{The first-order approximation:}
\begin{align}
    (\mb I-\alpha\mb L)^{1/2}\approx \sqrt{1-\alpha}\mb I+\frac{\alpha}{2\sqrt{1-\alpha}}\mb (\mb I-\mb L).\label{eq:approx1}
\end{align}
\textbf{The second-order approximation:}
\begin{align}
    (\mb I-\alpha\mb L)^{1/2}\approx&\sqrt{1-\alpha}\mb I+\frac{\alpha}{2\sqrt{1-\alpha}}(\mb I-\mb L)\nonumber-\frac{\alpha^2}{8(1-\alpha)^{3/2}}(\mb I-\mb L)^2.\label{eq:approx2}
\end{align}

\textbf{Derivation of approximation results:}
To use Chebyshev approximation, we first define the symmetric normalized adjacency $\mb A=\mb I-\mb L$, whose eigenvalues (denoted as $\sigma(\mb A)$) satisfy $\sigma(\mb A)\subset[-1,1]$. Then, we aim to approximate 
\begin{align}
    f(\mb A)=(\mb I-\alpha\mb L)^{1/2}=((1-\alpha)\mb I+\alpha\mb A)^{1/2}
\end{align}
using Chebyshev series.
Following GCN~\cite{jiang2019semi}, it is given by
\begin{align}
    f(\mb A)\approx\sum_{k=0}^K a_kT_k(\mb A),
\end{align}
where $T_k(\cdot)$ are Chebyshev polynomials of the first kind:
\begin{align}
    &T_0(t)=1,\quad T_1(t)=t,\\
    &T_{k+1}=2tT_k(t)-T_{k-1}(t),\quad k\ge 1.
\end{align}
And $a_k$ are Chebyshev coefficients defined by the cosine-series formula
\begin{align}
    a_0&=\frac{1}{\pi}\int_0^{\pi}f(\cos\theta)\,{\rm d}\theta\qquad\text{and}\\
    a_k&=\frac{2}{\pi}\int_0^{\pi}f(\cos\theta)\cos(k\theta)\,{\rm d}\theta,\quad k\ge 1.
\end{align}

Similar to the observations in GCN, we find that low-order approximations already perform quite well. Thus, we present the approximation results for $K=1$ and $K=2$ here:
\begin{align}
    f(\mb A)&\approx \sqrt{1-\alpha}\mb I+\frac{\alpha}{2\sqrt{1-\alpha}}\mb A,\qquad\text{if $K=1$}\\
    f(\mb A)&\approx \sqrt{1-\alpha}\mb I+\frac{\alpha}{2\sqrt{1-\alpha}}\mb A-\frac{\alpha^2}{8(1-\alpha)^{3/2}}\mb A^2,\qquad\text{if $K=2$}.
\end{align}

\section{Experimental Settings}

\subsection{Datasets}\label{sec:dataset}

\begin{table}
    \caption{Statistics of the preprocessed datasets.}
    \vspace{-0.5em}
    \centering
    \begin{tabular}{ccccc}\toprule
         Dataset & \#Users & \#Items & \#Interactions & \#Sparsity\\
         \midrule
         \texttt{Yelp} & 236,482 & 128,727 & 3,590,135 & 99.98\%\\
         \texttt{Amazon Sports} & 103,757 & 201,957 & 1,691,350 & 99.99\%\\
         \texttt{Amazon CDs} & 7,685 & 5,841 & 61,564 & 99.86\%\\
         \bottomrule
    \end{tabular}
    \label{tab:dataset}
\end{table}

We adopt three real-world datasets for verification, i.e., \textbf{Yelp}~\cite{liu2024llm,liu2025llmemb,yang2025darec}, \textbf{Amazon Sports}~\cite{xu2025slmrec,huang2024improving,lin2025rec}, and \textbf{Amazon CDs}~\cite{dai2023uncovering,bao2024decoding,zhang2025reinforced}.
Yelp\footnote{\url{https://business.yelp.com/data/resources/open-dataset}} includes a large number of check-in records, which can be used for point-of-interest recommendation. The Amazon\footnote{\url{https://cseweb.ucsd.edu/~jmcauley/datasets.html}} datasets are collected from an e-commerce platform. Sports and CDs are two sub-categories of this dataset.
We split the training and test datasets using the same method as in~\cite {liu2024llm}. Specifically, we first filter out users and items with fewer than 10 interactions for Amazon and fewer than 5 interactions for Yelp. Then, we split the remaining data chronologically into training, validation, and test sets in an 8:1:1 ratio. The statistics of the preprocessed datasets are summarized in Table~\ref{tab:dataset}.

\subsection{Baselines}\label{app:baseline}
To demonstrate the superiority of our method, we compare it with four state-of-the-art LLM-enhanced sequential recommendation methods that utilise LLMs to augment conventional sequential recommendation models.
\begin{itemize}[leftmargin=*]
    \item \textbf{LLMInit}~\cite{zhang2025llminit}. It integrates pretrained LLM embeddings into conventional collaborative filtering models through selective initialization strategies, including random, uniform, and variance-based selections.
    \item \textbf{LLMEmb}~\cite{liu2025llmemb}. It proposes a Supervised Contrastive Fine-Tuning approach, along with attribute-level data augmentation. It also includes Recommendation Adaptation Training to emphasize the importance of integrating collaborative signals into LLM-generated embeddings.
    \item \textbf{LLM-ESR}~\cite{liu2024llm}. It designs a dual-view modelling framework to combine LLM and conventional SRS. It also proposes a retrieval-augmented self-distillation method to enhance user preference representation by utilising more informative interactions from similar users.
    \item \textbf{LLM2Rec}~\cite{he2025llm2rec}. It proposes a two-stage training framework that first adapts LLMs to infer item relationships based on historical interactions. Followed by refining LLMs into structured item embedding models that encode both semantic and collaborative information.
\end{itemize}

\subsection{Backbones}\label{app:backbone}
To show the flexibility of our enhancement method, we conduct experiments on three well-known sequential recommendation models. The primary distinction between these models lies in the sequence encoder.
\begin{itemize}[leftmargin=*]
    \item \textbf{GRU4Rec}~\cite{hidasi2015session}. It adapts the GRU module as the sequence encoder.
    \item \textbf{Bert4Rec}~\cite{sun2019bert4rec}. Inspired by the training pattern of BERT~\cite{devlin2019bert}, Bert4Rec proposes a cloze task that masks a proportion of items in a single sequence. The sequence encoder of Bert4Rec is the stack of bi-directional self-attention layers.
    \item \textbf{SASRec}~\cite{kang2018self}. Compared to Bert4rec, SASRec employs the causal self-attention layer as the fundamental unit of its sequence encoder.
\end{itemize}
Note that among these three models, Bert4Rec employs the Cross-Entropy loss, while the other models utilize BPR loss. Since numerous studies~\cite{sun2019bert4rec,xu2024understanding,klenitskiy2023turning} have demonstrated that the Cross-Entropy loss outperforms BPR loss in these sequential recommendation models, and the relative performance of models is likely to shift when identical loss functions are applied. Therefore, to ensure experimental fairness, we consistently applied Cross-Entropy loss across all three models.

\subsection{Implementation Details}\label{app:implement}

We implement all methods using PyTorch and employ Adam as the optimizer in all experiments. For all LLM-enhanced methods, we use \llmname{Qwen3-Embedding-8B} as the backbone LLM for encoding text to ensure fairness. For all sequential recommendation backbones, we tune their hyperparameters to an optimal. For our method, the weight decay is selected from \{1e-3, 1e-5, 0\}. The learning rate is fixed to 1e-3. $\alpha$ is selected from \{0.1, 0.2, 0.3, 0.4, 0.5\}. In constructing the item-item co-occurrence graph, we apply K-sparsification~\cite{he2020lightgcn,peng2024powerful} and retain only $K$ neighbors with the highest co-occurrence frequency for each node. We select $K$ from $\{1,3,5,7,9\}$. We set the total number of training epochs to 100. We also conduct early stopping based on the NDCG@10 on the validation set, stopping training when the NDCG@10 doesn't increase for 10 consecutive epochs. All hyperparameters of the compared baselines are tuned to ensure optimal performance.

\section{More Experimental Results}

\begin{figure*}
\centering
  \includegraphics[width=\linewidth]{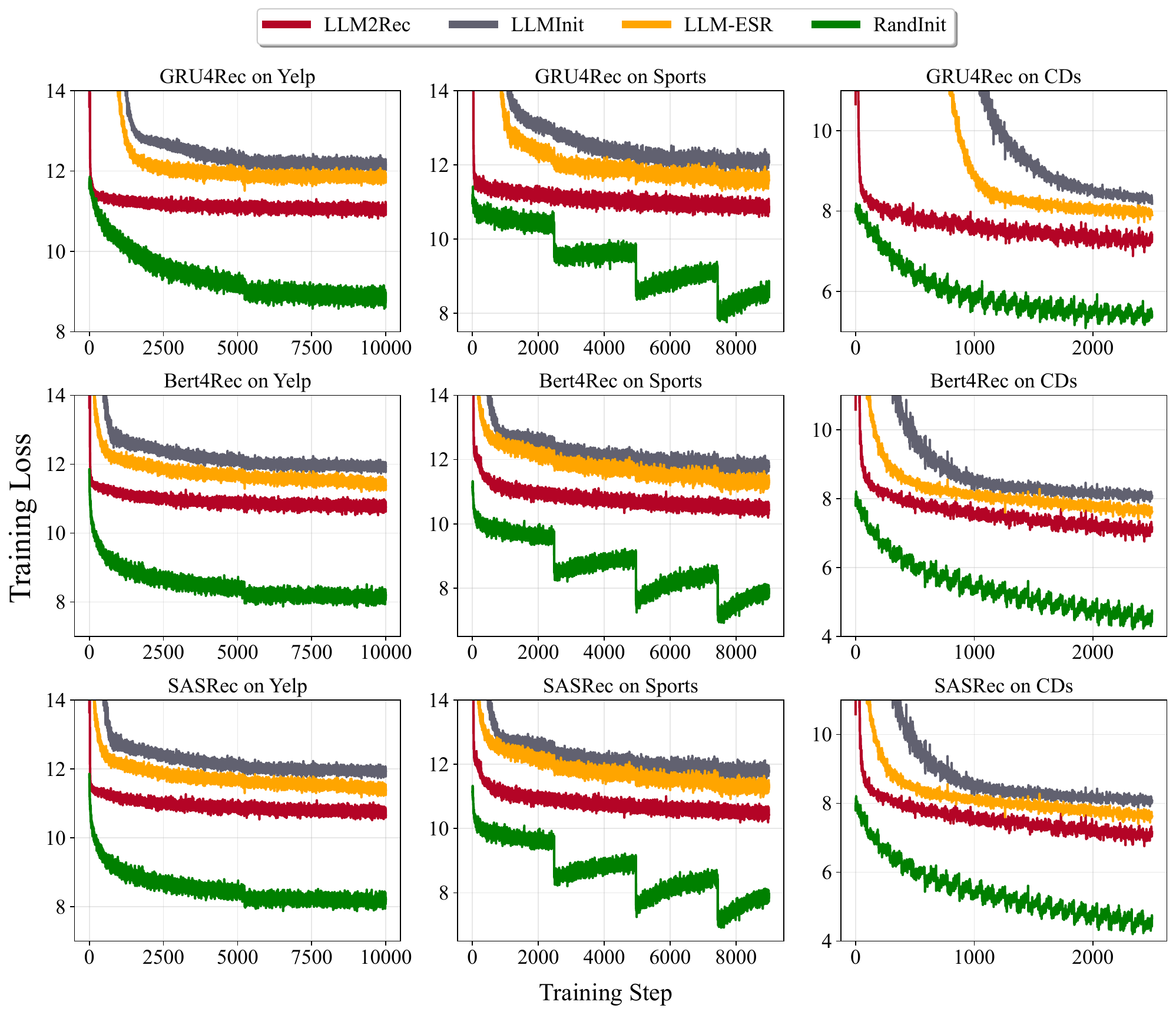}
    \vspace{-1.5em}
  \caption{Training loss of several methods, including the standard randomly initialized model (denoted as RandInit) and state-of-the-art LLM-enhanced methods.}
  \label{fig:loss_methods_all}
    \vspace{-0.5em}
\end{figure*}

\begin{figure*}
\centering
  \includegraphics[width=\linewidth]{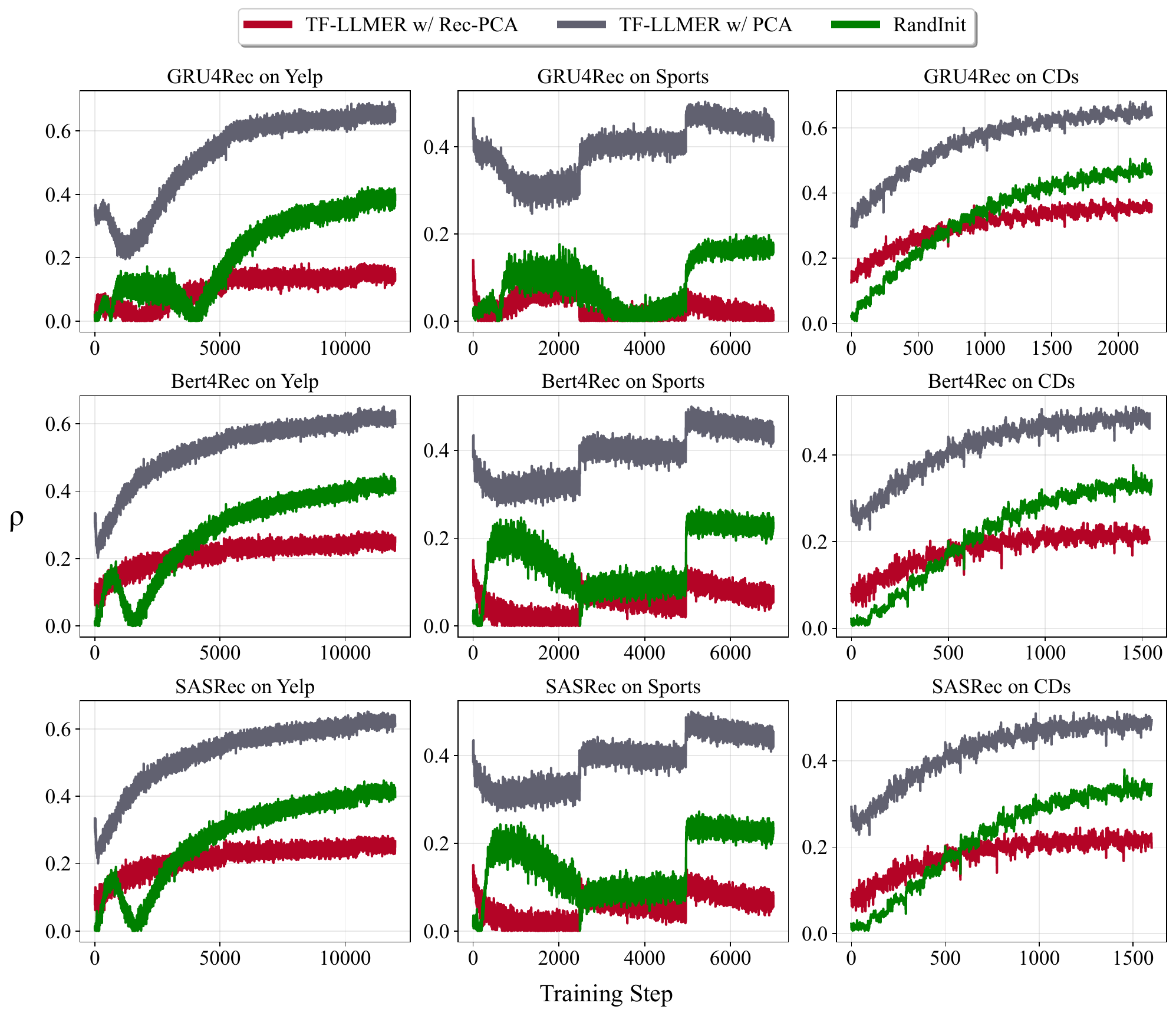}
  \caption{Training curves of $\rho$ on the training set for Rec-PCA, our method with vanilla PCA, and random initialization (denoted as RandInit) across three datasets (Yelp, Sports, and CDs) and three backbones.}
  \label{fig:rho_all}
\end{figure*}

\begin{figure*}
\centering
  \includegraphics[width=\linewidth]{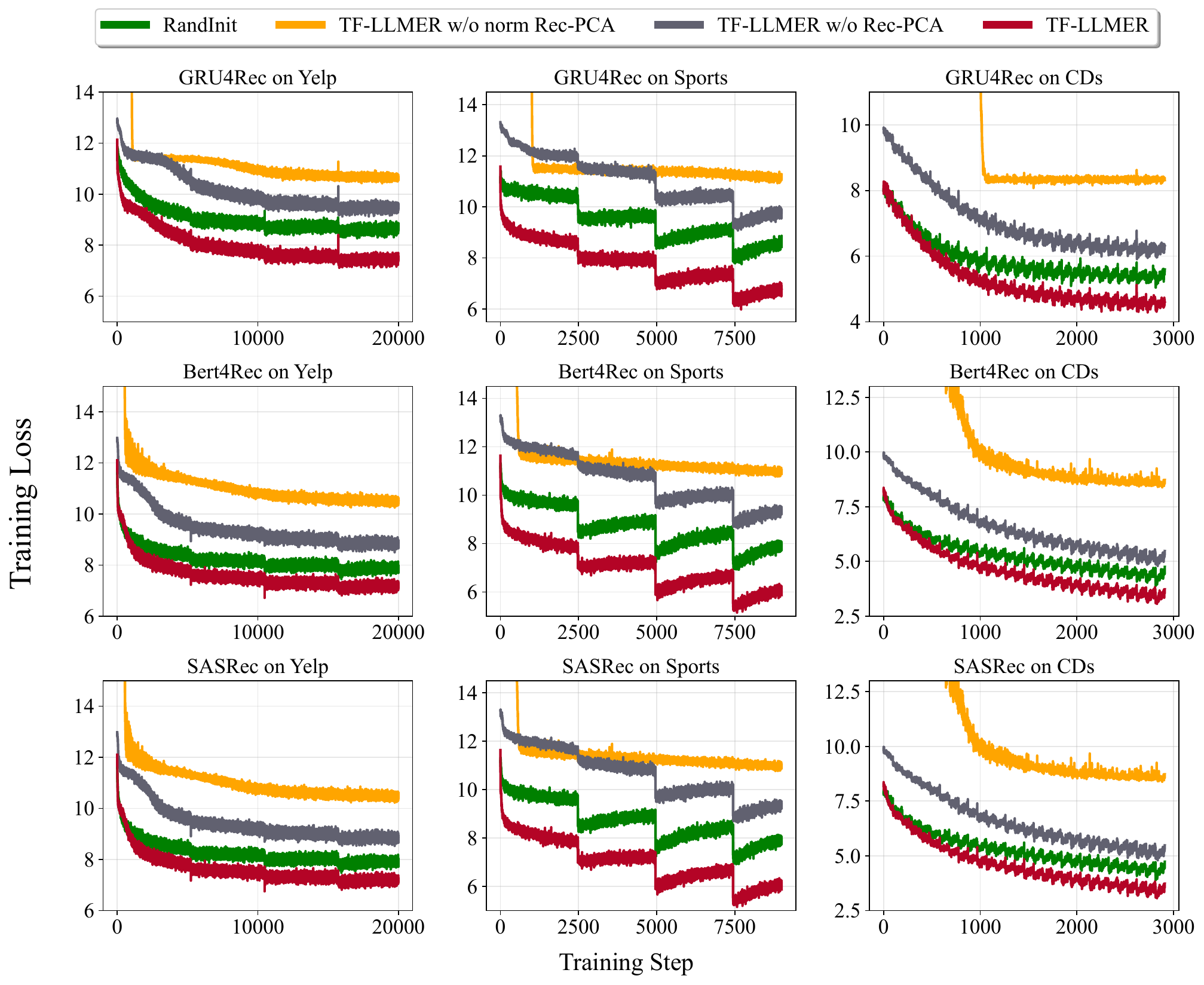}
    \vspace{-1.5em}
  \caption{Training loss of several variants, including randomly initialized models (denoted as RandInit), our method without normalization and Rec-PCA (denoted as \ourmethod w/o norm Rec-PCA), our method without Rec-PCA (denoted as \ourmethod w/o Rec-PCA), and the full version of our method (denoted as \ourmethod).}
  \label{fig:loss_all}
    \vspace{-0.5em}
\end{figure*}

\subsection{Training Losses of Several LLM-Enhanced Recommendation methods}\label{app:loss_method_all}

In this section, we present the complete training loss results for several strong LLM-enhanced methods. We also report the training loss of randomly initialized recommendation models (denoted as RandInit) for comparison. The results are provided in Figure~\ref{fig:loss_methods_all}. The results demonstrate that, compared to the original randomly initialized model, the training loss of all state-of-the-art LLM-enhanced methods ceases to decrease rapidly and stabilizes at significantly higher levels. It indicates that all LLM-enhanced methods hinder the training of sequential recommendation models.

\subsection{Training Curve of $\rho$ for Compared Methods}

This section provides the complete results for the training curves of $\rho$ for the compared methods, including our Rec-PCA, our method with standard PCA, and the randomly initialized recommendation model (denoted as RandInit), across all three datasets. The results are provided in Figure~ƒ\ref{fig:rho_all}. The results demonstrate that across all datasets, our Rec-PCA consistently reduces effective coherence. This can further reduce the upper bounds of $\kappa(\hat{\mb E}_\mc U\hat{\mb E}_\mc U^\top)$ and $\kappa(\mb H_{h\mc U})$, thereby providing a theoretical basis for promoting training.

\subsection{Losses of Several Variants of Our Method}

This section presents the training loss for all variants of our method: our method without embedding normalization and Rec-PCA (\ourmethod w/o norm Rec-PCA), our method without Rec-PCA (\ourmethod w/o Rec-PCA), and the full version of our method (\ourmethod). We also present the training loss of randomly initialized recommendation models for comparison. The results are given in Figure~\ref{fig:loss_all}. From the results, we observe that:
\begin{enumerate}
    \item Without embedding normalization, the loss starts extremely high. The decline of training loss stalls rapidly and finally stabilizes at a level substantially higher than that achieved with normalization.
    \item After replacing vanilla PCA with Rec-PCA, the training loss further decreases. Throughout the entire training process, the loss remains lower than that of other comparison methods. The training loss is also significantly lower than RandInit.
\end{enumerate}
The above observations validate that our proposed item embedding normalization and Rec-PCA are effective in facilitating training LLM-enhanced methods.

\subsection{Ablation Study}\label{app:abl}

\begin{figure}
\centering
  \includegraphics[width=\linewidth]{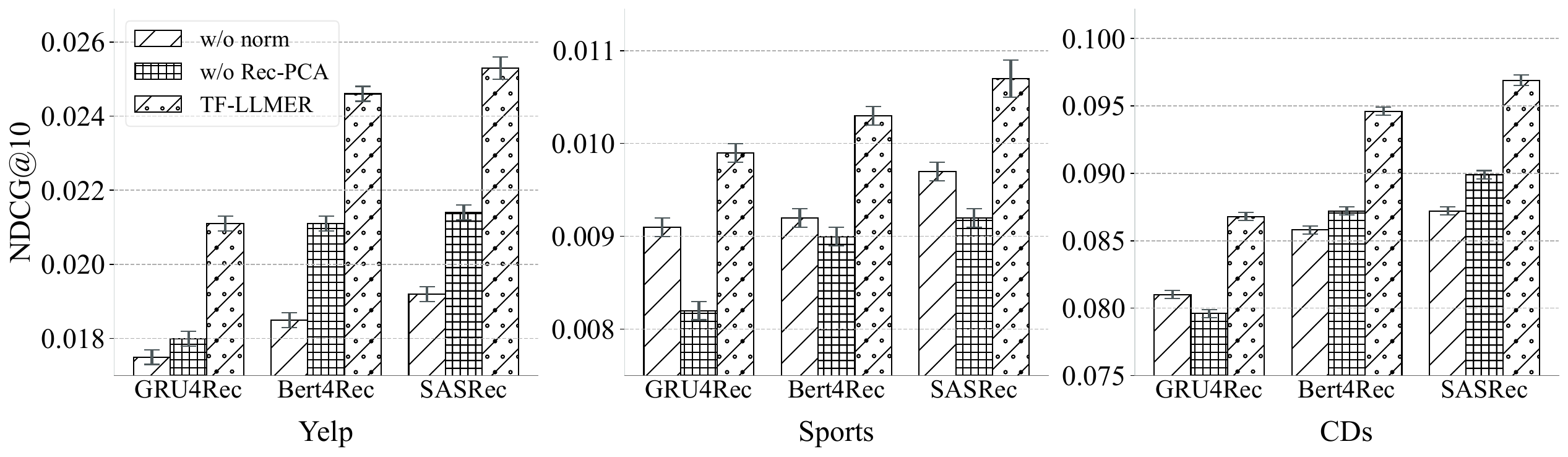}
  \caption{Ablation study of our method, including three variants.}
  \label{fig:abl_all}
\end{figure}

Figure~\ref{fig:abl_all} provides the full results of the ablation study on all three datasets. Across all three datasets, we observe that all variants are inferior to the original \ourmethod, demonstrating the effectiveness of the key components. Moreover, we observe that item normalization has a large effect on Yelp and Amazon Sports, while Rec-PCA has a greater impact on Amazon Sports. We believe this relates to the characteristics of the item representations generated by the LLM across different datasets. On Yelp and Amazon CDs, the LLM-generated item representations exhibit larger magnitude differences, leading to more significant training challenges. On Amazon Sports, increasing the distinguishability between effective items significantly improves training effectiveness, thereby amplifying Rec-PCA's impact on recommendation performance.

\subsection{Hyperparameter Analysis}\label{sec:hyper}

\begin{figure}
\centering
  \includegraphics[width=0.7\linewidth]{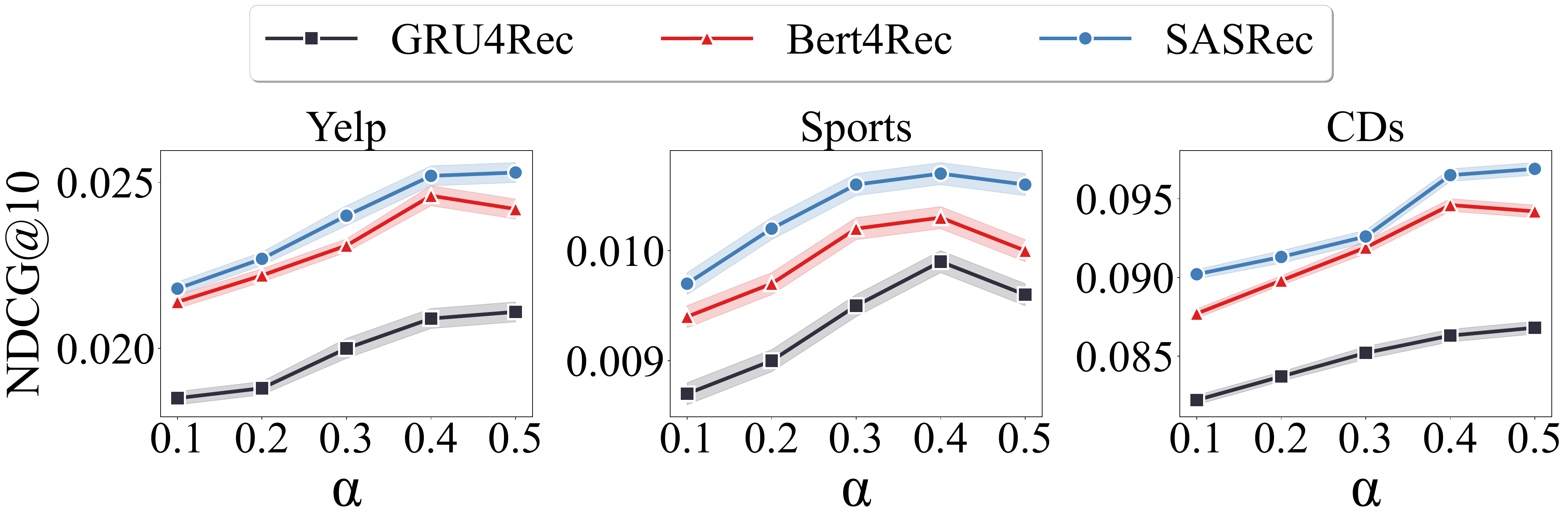}
  \caption{Hyperparameter $\alpha$.}
  \label{fig:hyper_alpha}
\end{figure}

\begin{figure}
\centering
  \includegraphics[width=0.7\linewidth]{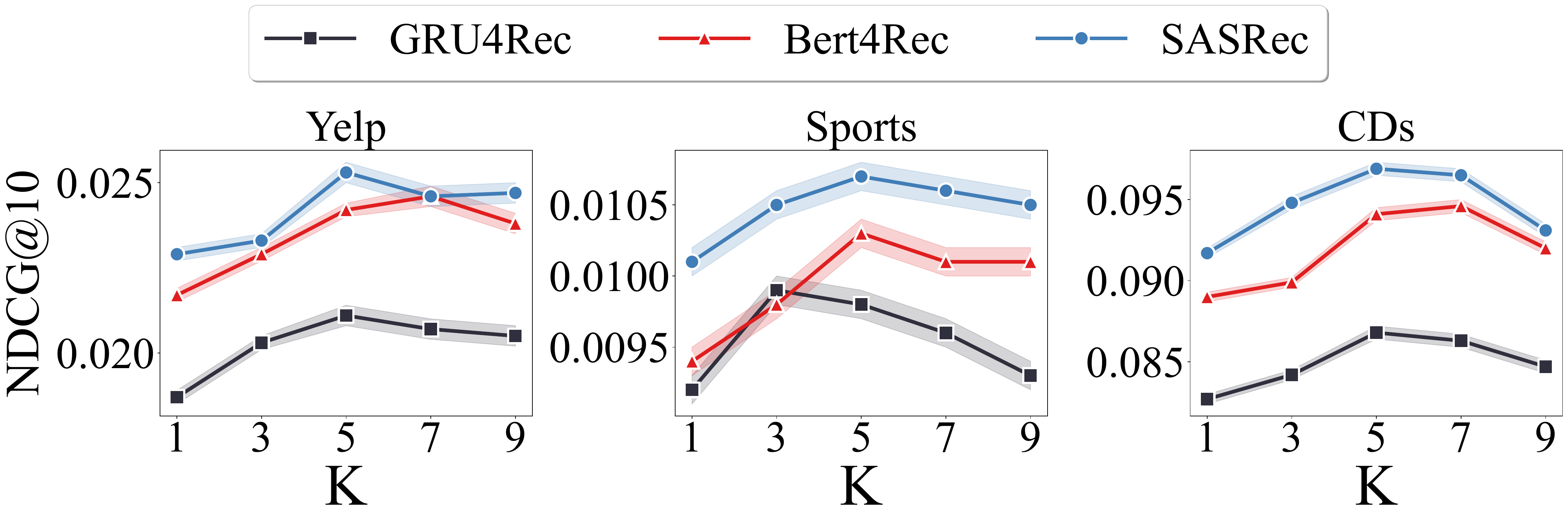}
  \caption{Hyperparameter $K$.}
  \label{fig:hyper_k}
\end{figure}

\textbf{Effects of $\alpha$.}
In Eq.(\ref{eq:target}), we use $\alpha$ to balance the objectives of maximizing variance and maintaining a collaborative relationship in Rec-PCA. Note that to ensure that our maximization objective in Eq.(\ref{eq:target}) is diagonalizable, we require $\alpha\le 0.5$. In Figure~\ref{fig:hyper_alpha}, we report the effect of $\alpha$. The results suggest that the trends across different datasets and backbones are similar. For general, the best choice of $\alpha$ lies in $\{0.4, 0.5\}$.

\textbf{Effects of $K$.}
In Rec-PCA, we use the item-item co-occurrence graph to capture the collaborative relationships among items in the recommendation task. Since the original graph contained significant noise, we apply K-sparsification, as is standard practice in graph-based recommendation systems~\cite{he2020lightgcn,peng2024powerful}. This process retained only the $K$ neighbors with the highest co-occurrence frequency for each node. In Figure~\ref{fig:hyper_k}, we analyze the effect of $K$. The best choice of $K$ lies in $\{3,5,7\}$. We find that both too large and too small values of $K$ lead to worse performance.


\end{document}